\documentclass{article}
\usepackage{amsmath,graphicx}

\usepackage{amsthm}
\usepackage{amssymb}
\usepackage{xypic}
\usepackage{hyperref}
\theoremstyle{plain}
\newtheorem{theorem}{Theorem}
\newtheorem{proposition}{Proposition}
\newtheorem{corollary}{Corollary}
\newtheorem{lemma}{Lemma}
\theoremstyle{definition}
\newtheorem{definition}{Definition}

\newtheorem{example}{Example}

\def\image{\textrm{image }}

\def\id{\textrm{id}}

\def\cat{\mathbf}

\title{Constant rank factorizations of smooth maps}
\author{Michael Robinson\\
Mathematics and Statistics\\
American University\\
Washington, DC, USA\\
michaelr@american.edu}

\begin{document}
\maketitle

\begin{abstract}
Sonar systems are frequently used to classify objects at a distance by using the structure of the echoes of acoustic waves as a proxy for the object's shape and composition.  Traditional synthetic aperture processing is highly effective in solving classification problems when the conditions are favorable but relies on accurate knowledge of the sensor's trajectory relative to the object being measured.  This article provides several new theoretical tools that decouple object classification performance from trajectory estimation in synthetic aperture sonar processing.  The key insight is that decoupling the trajectory from classification-relevant information involves factoring a function into the composition of two functions.
The article presents several new general topological invariants for smooth functions based upon their factorizations over function composition.  These invariants specialize to the case when a sonar platform trajectory is deformed by a non-small perturbation.  The mathematical results exhibited in this article apply well beyond sonar classification problems.  This article is written in a way that supports full mathematical generality.  
\end{abstract}

\tableofcontents

\section{Introduction}
\label{sec:intro}

Sonar systems are frequently used to classify objects at a distance by using the structure of the echoes of acoustic waves as a proxy for the object's shape and composition.  To obtain good classification accuracy, many sonar systems use a moving sensor platform to produce a \emph{synthetic aperture}.  As the sensor platform's position changes, one can measure how the echoes change, and from these changes deduce properties of the object.

Traditional synthetic aperture processing is highly effective in solving classification problems when the conditions are favorable but relies on accurate knowledge of the sensor's trajectory relative to the object being measured.  Any deviations from the expected trajectory degrade the classification performance of the overall system.  Because of this, there is a well-established practice of \emph{motion compensation} and \emph{autofocus} algorithms that iteratively apply corrections to the estimated trajectory in hopes of removing these deviations.

While the autofocus approach works well for small perturbations of the expected trajectory, it suffers when the initial guess of the trajectory is bad.   It is natural to ask, ``can one obtain good classification accuracy from synthetic aperture sonar images \emph{without} a good trajectory estimate?''
Given that the signatures (the ensemble of all received echoes) of different targets appear to have quite different geometry and topology \cite{sonarspace}, classification may be possible even given inaccurate trajectory information.

This article provides several new theoretical tools that decouple object classification performance from trajectory estimation in synthetic aperture sonar processing.  
It does so by defining several new general topological invariants for smooth functions, which specialize to the case when a sonar platform's trajectory is deformed by a non-small perturbation.

Although this article is theoretical in nature and does not seek to provide ready-to-use processing algorithms,
the decoupling of trajectory from classification performance means that good trajectory information is not necessary for sonar object classification.
Indeed, Proposition \ref{prop:identical_images} establishes that the image of the sonar signature---the space of pulse echoes without regard to the order in which they arrive---is the governing factor behind the good classification performance observed in \cite{sonarspace}.

The mathematical results exhibited in this article apply well beyond sonar classification problems, so this article is written in a way that supports full mathematical generality.  The key insight is that decoupling the trajectory from classification-relevant information involves factoring a function into the composition of two functions.
Generalizing this insight, for a smooth function $u$ we define two categories $\cat{QuasiP}(u)$ and $\cat{Const}(u)$ that describe classes of factorizations of $u$ by function composition under certain constraints.  While the constraints were motivated by the needs of sonar classification, they are sufficiently weak to permit wide usage of these categories.

The plan for the paper is as follows.  In the next few subsections, we discuss the literature for synthetic aperture sonar target classification (Section \ref{sec:historical}), which provides context for the contributions of this article (Section \ref{sec:contribution}).  In Section \ref{sec:motivation}, we introduce the concept of circular synthetic aperture sonar (CSAS) and the specific classification problem to be addressed in the rest of the article.  In Section \ref{sec:definitions}, we introduce the factorization categories $\cat{QuasiP}(u)$ and $\cat{Const}(u)$ and prove several fundamental results about them, including a complete characterization of $\cat{QuasiP}(u)$ for CSAS.  Although this complete characterization is of limited utility for $\cat{Const}(u)$, we show that a pair of views of the same target establishes an equivalence between functions using common factors.  We develop this idea in Section \ref{sec:comparing} by defining a category $\cat{CRSE}(u_1,u_2)$ of factorization equivalences, and ultimately prove a characterization theorem for this category.  The theoretical machinery we have by then constructed is used to revisit the motivating CSAS example in Section \ref{sec:application}.  Finally, we provide a brief conclusion in Section \ref{sec:conclusion}. 

\subsection{Historical context}
\label{sec:historical}

The process of forming an image from synthetic aperture sonar data uses a bank of spatial matched filters (pixels) over the region where a sonar target is located.  Because sonar targets tend to be spatially localized, high-resolution image-based methods are highly effective at rejecting background clutter.  Image-based methods generally require accurate knowledge of the sensor platform's trajectory relative to the object being measured, careful control of the sonar waveform, and intimate knowledge of the clutter environment.  Hundreds of papers attest to the effectiveness of image-based methods when the conditions are favorable.  As a brief sample, the reader is encouraged to consult \cite{Bucaro_2012, Bucaro_2010,Chambers_2007,Gaumond_2006,Carin_2006, Dasgupta_2005,Liu_2004, Gruber_2004, Carin_2004}, though this list is not exhaustive.

When information about the sensor platform trajectory is not sufficiently accurate, the recovered energy spreads across neighboring filters in the image (pixels), leading to blurring.  Theoretically, blurring is a manifestation of the instability of the Fourier transform when the domain is deformed \cite{Hormander_1971}.  Classification becomes difficult if one is forced to start with blurred images.

Since the sensor platform has inertial mass, it is often possible to determine compensations to the filters to account for its motion.  When motion compensations are applied, this results in a more focused image \cite{Chevillon1998AFM}.  Moreover, when the trajectory is known to be close to a given trajectory---a circular orbit around the target for instance---then motion compensation can improve object classification \cite{Marston2021SpatiallyVA}.  The impacts of motion blurring on classification with and without compensation have been discussed in the literature, for instance \cite{Cotter2018TrackingAC,Bonifant1999AnAO}.

Because the \emph{scattering transform} is Lipschitz continuous relative to small deformations \cite{Andn2014DeepSS,Bruna2012InvariantSC,Mallat2011GroupIS}, it can be used to mitigate small, arbitrary trajectory and propagation distortions in sonar signals \cite{Saito2017UnderwaterOC}.  The scattering transform works by decomposing the signal via local convolutions with wavelets or another frame of localized functions.  When the distortions are small, the fact that sonar signals of interest tend to be sparse in this basis can be used to isolate the target's response from the distortions.  Robust stability results for many frames commonly in use have been established \cite{Wiatowski2015DeepCN}.  However, the scattering transform suffers from the ``curse of dimensionality,'' which means that it can be rather difficult to determine \emph{a priori} which subspace constitutes the target response.  Consequently, the scattering transform can be computationally expensive.  This happens because the scattering transform explicitly characterizes the distinction between target responses and trajectory distortions.  Since the present article explores this distinction implicitly, it provides a perspective on sonar data that is complementary to the scattering transform, and may ultimately require less computation.

The fact that classification can sometimes succeed with poor trajectory information suggests that non-image based methods could succeed as well.  Statistical machine learning methods applied to the raw sonar echoes are a natural choice.  Indeed, machine learning has been successfully applied to sonar classification problems over the past three decades, going back to the beginnings of the subject \cite{G2020SurveyOD,Chen2016DeepNN,Li_2004,Azimi_2000,Yao_1999,Carpenter_1998,Azimi_1998,Gorman_1988}.  Unfortunately, machine learning requires large and diverse training data sets, which can be expensive to obtain.  The implication for sonar classification is that the data need to contain not only the targets of interest, but also enough diversity to capture physical effects due to shadows, background structure, and pose \cite{Williams2019DemystifyingDC}.

A different perspective is to take the \emph{space of echoes} arising from a given target \emph{as a feature} in the abstract, rather than as a point in some high dimensional space.  Comparing such abstract spaces using \emph{homotopy equivalence} is the dominant technique in the mathematical discipline of algebraic topology.  By construction, homotopy equivalence is automatically robust to deformations, so it might seem to be ideally suited for sonar classification problems.  Unfortunately homotopy equivalence is neither easily testable nor robust to statistical noise.  However, a regularized version of homotopy equivalence, called \emph{persistent homology}, is algorithmically computable and is noise robust.  The use of persistent homology to study data sets forms the basis of topological data analysis (TDA), for which a standard classification pipeline has emerged \cite{Chen2021ApproximationAF,Chen2019ATR,Li2017MetricsFC}.

TDA can be applied to sonar classification by comparing the space of echoes arising from an unknown target with a dictionary of spaces for known targets.  The traditional TDA pipeline does not consider the ordering of the pulses, since persistent homology is only sensitive to the relationship between similar echoes.  Those pulses that are transmitted from nearby locations will be near one another in the space of echoes, and so also will those pulses arising from symmetries in the target.  While most sonar targets are unlikely to be actually symmetric, strong spatial Fourier modes can yield the same effect.  There is a large literature on signals with strong Fourier modes, which are called \emph{almost periodic} functions. 

Unfortunately, homotopy equivalence is not the proper equivalence for sonar targets.  Since the standard TDA pipeline cannot discriminate between spaces that are homotopy equivalent, TDA is not strictly applicable to sonar echoes.  Specifically, homotopy equivalence is both too weak and too strong!  As a brief example of the inappropriateness of homotopy equivalence in sonar, it has been shown that homotopy equivalence spuriously detects the direction that the sensor platform orbits a target in CSAS \cite{Robinson_RSN}.  Conversely, homotopy equivalence fails to discriminate between the echoes arising from a spherical target and an extended one, such as a rod, since both have contractible spaces of echoes.  Nevertheless, persistent homology can yield an effective classifier for sonar targets \cite{sonarspace}.  Moreover, justification for using the space of echoes is established by Proposition \ref{prop:identical_images}.

To obtain a better equivalence, one needs to supply information about the ordering of the pulses.  In this case, the object to be classified is not merely the space of echoes, but a function from the trajectory to the space of echoes.  There are some TDA tools that look at such functions \cite{Takeuchi2021ThePH,Harker2014InducingAM,Dey2014ComputingTP}; these are special cases of what this article presents in Section \ref{sec:circular}.

This article relies on the idea of factoring functions through various manifolds.  For the practical usage of factorizations in CSAS, one must factor through a circle.  Such factorizations through circles have recently become practical through the moniker of \emph{finding circular coordinates} for data sets using \emph{persistent cohomology}.  The interested reader is referred to \cite{Luo2020GeneralizedPF,Perea2018SparseCC,Perea2018MultiscalePC,VejdemoJohansson2015CohomologicalLO,DeSilva_2011}.

At a high level, this article seeks to explore a kind of signal processing that is aware of the underlying topology of the space of echoes.  We therefore seek a notion of a \emph{topological filter} \cite{Robinson_2014}.  Several kinds of novel topological filters are known \cite{Essl2020TopologicalIF, Robinson_qplpf, Robinson_SampTA_2015}, though it is clear that these do not exhaust the possibilities.
The results of this article could provide a basis for extending the above works, by supplying the base topological data upon which topological filters are built.

\subsection{Contributions}
\label{sec:contribution}

From a mathematical perspective, the focus of this article is to classify smooth maps $u: M \to N$ based on factorizations of $u$ into the composition of two functions.   The article defines two new differential topological invariants for a smooth map $u$ that are distinct from its homotopy class.  These invariants are the categories $\cat{Const}(u)$ and $\cat{QuasiP}(u)$ of factorizations of $u$ (Definition \ref{def:constant_rank}).
In addition to the factorization categories for a single map, the article also defines a new differential topological invariant $\cat{CRSE}(u_1,u_2)$ for a \emph{pair} of smooth maps $u_1$, $u_2$ that relates their respective factorizations to each other (Definition \ref{def:crse}).  Each isomorphism class of $\cat{CRSE}(u_1,u_2)$ identifies an individual sonar target that could have yielded both the responses $u_1$ and $u_2$ under different conditions.

From a practical perspective, the three new invariants allow one to deduce whether changes in sensor configuration or trajectory might be able to confound two objects within a CSAS collection.  (Unfortunately, we have not discovered an efficient way to compute these invariants in all cases, though $\cat{QuasiP}(u)$ can be computed algorithmically in the case of CSAS using Theorem \ref{thm:circle_quasip}.)

As already noted, \cite{sonarspace} shows that persistent homology can be used to classify targets without image formation.  But why should a topological invariant (homology) be useful in classification?  Proposition \ref{prop:identical_images} establishes a sufficient condition for homology of the space of echoes to be an effective solution.  Moreover, Theorem \ref{thm:circle_quasip} provides a more complete answer to this question by establishing that the $\cat{QuasiP}(u)$ of a CSAS signature is determined by the fundamental group (a topological property), and this is preserved if the trajectory is deformed.  This is still only a partial answer for two reasons: (1) not only topology but also geometry impacts persistent homology, and (2) many realistic targets have CSAS signatures with trivial $\cat{QuasiP}(u)$, limiting its usefulness.  Although $\cat{Const}(u)$ does not share this second limitation, we have not succeeded in discovering an algorithmically computable characterization of it.
 
This article establishes the following new results:
\begin{enumerate}
\item The factorization categories are functorial (Propositions \ref{prop:left_functor}, \ref{prop:right_functor}, and \ref{prop:crse_functoriality}); changing the domain or codomain induces functors on the factorization categories,
\item The factorization categories are related to, but distinct from, homotopy classes (Corollary \ref{cor:circle_homotopy_quasip} and Proposition \ref{prop:crse_homotopy}),
\item When maps are related via a diffeomorphism, their respective categories are isomorphic (Theorem \ref{thm:crse_diffeo}) and the image of $\cat{CRSE}$ in them is maximal (Theorem \ref{thm:surjectivity}),
\item When the same target appears in multiple settings, the factorization categories of the resulting signatures have equivalent subcategories based on the target (Theorems \ref{thm:coslice_equivalent} and \ref{thm:crse_coslice}), and
\item When the domain $M$ is a circle $S^1$ (directly relevant to CSAS), the quasiperiodic factorization category is completely classified by subgroups of the fundamental group of $N$ (Theorem \ref{thm:circle_quasip}).
\end{enumerate}

\section{Motivating example: circular synthetic aperture sonar (CSAS)}
\label{sec:motivation}

Consider a circular synthetic aperture sonar (CSAS) collection in which the sonar platform orbits a fixed target, as shown in Figure \ref{fig:collection_geometry}.  For the purposes of this article, suppose that the sensor's \emph{standoff range} $R$ is constant, and that its \emph{look angle} $\theta$ increases monotonically (but not necessarily linearly) from $0^\circ$ to $360^\circ$ over the length of the collection.  If the target's scatterers are sufficiently reflective when compared to the receiver noise and clutter in the scene, then determining the ranges to them is not a difficult problem.  This fact means that we may safely assume that the standoff range $R$ is fixed and constant without impacting the realism of the model.  If the actual range varies, this may be compensated for by applying a $\theta$-dependent modulation to the raw echoes.  

We will follow mathematical tradition and denote the set of angles by the circle $S^1$, so that $\theta \in S^1$ specifies the look angle.  The sensor emits \emph{sonar pulses} as it moves along its trajectory and records their echoes from the target.  We will neglect receive and transmit filter effects for this article to keep the exposition simple, though filters can be incorporated easily using the framework the article develops.  Let us therefore assume that the transmitted pulses are impulses, and the sensor's receiver covers a wide band of frequencies.  

\begin{figure}
\begin{center}
\includegraphics[height=2.25in]{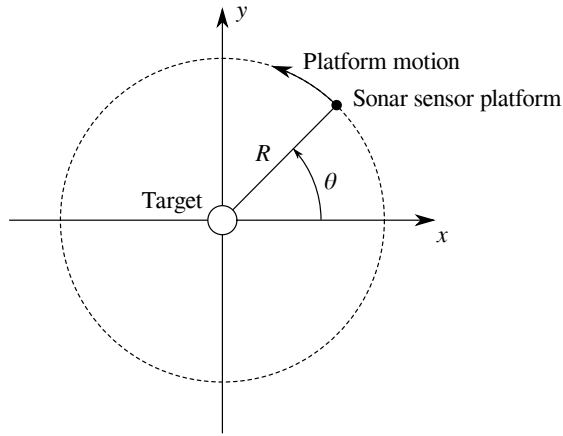}
\caption{General CSAS collection geometry wherein the sensor platform orbits a target.}
\label{fig:collection_geometry}
\end{center}
\end{figure}

Given the standoff range $R$ and the wave phase speed $c$, the following equation models the \emph{signature} $s$, which describes the received echoes as a function of the look angle $\theta$ and received frequency $f$ received from a set of $P$ point scatterers,
\begin{equation}
  \label{eq:general_point_scatterers}
  s(\theta,f) = \sum_{p=1}^P \frac{a_p \exp\left(-\frac{2 \pi i f}{c} r_p \cos (\alpha_p-\theta)\right)}{\sqrt{\left(r_p \sin (\alpha_p-\theta)\right)^2 + \left(R + r_p \cos (\alpha_p-\theta)\right)^2}}
\end{equation}
in which each point scatterer is parameterized by its complex reflectivity $a_p$, its distance $r_p$ to the target center, and its angle $\alpha_p$ relative to the $\theta=0$ axis.  This is, of course, a rather simplified model.  Substantially more realistic models (not just superpositions of point scatterers) are supported by the framework developed in the latter sections of this article.  Although we use a single look angle $\theta$ in this section, this is only to ensure that the signals can be represented ``on paper'' as images.  The framework supports ``look angles'' that can be represented as points on an arbitrary connected smooth manifold.

Equation \ref{eq:general_point_scatterers} sets up a correspondence between the look angle $\theta$ and the received echo as a function of pulse frequency $f$.  If the sensor platform's angular position with respect to the target is not known accurately, this function becomes the composition of $s$ with an unknown function of time $\theta = \phi(t)$, where $t$ is the time of pulse transmission.  Due to the inertia of the sensor platform, we may assume that the function $\phi$ is smooth.

If $\phi$ is an affine function, then traditional imaging methods (for instance, polar format or back projection \cite{Quegan}) provide a complete solution to the problem of determining the location and reflectivity of each scatterer.  These methods do not work at all if $\phi$ is an unknown smooth function that deviates far from an affine function.

To fix ideas, let us restrict attention to the case where the point scatterers in this target model can be grouped into two subsets: one subset has a $180^\circ$ rotational symmetry, while the other subset has a $120^\circ$ rotational symmetry.  We will usually refer to these symmetries as a \emph{$2$-fold symmetry} and a \emph{$3$-fold symmetry}, respectively.  One easy way that this may happen is if the first subset consists of exactly two point scatterers, placed opposite the target center, as in Figure \ref{fig:scatterer_locations}(a).  Similarly, the other subset could be realized as a subset of exactly three point scatterers, evenly spaced around a circle concentric with the target center, as in Figure \ref{fig:scatterer_locations}(b).  It is immediately clear that there are other possible ways to realize these kinds of symmetries using more point scatterers, and while this will be developed more completely in Proposition \ref{prop:signal_knot} (and following), these two simplistic models will suffice for the moment.  We will call these subsets of the set of point scatterers \emph{composite scatterers} for brevity.  Notice that each composite scatterer additionally has \emph{a reference rotation angle} $\alpha$ with respect to the $\theta=0$ axis.  The two composite scatterers will typically have different reference rotation angles, which means that the superposition of the two composite scatterers may not have any rotational symmetries at all.

\subsection{A composite scatterer model: quasiperiodic factorizations}

\begin{figure}
\begin{center}
\includegraphics[height=2.25in]{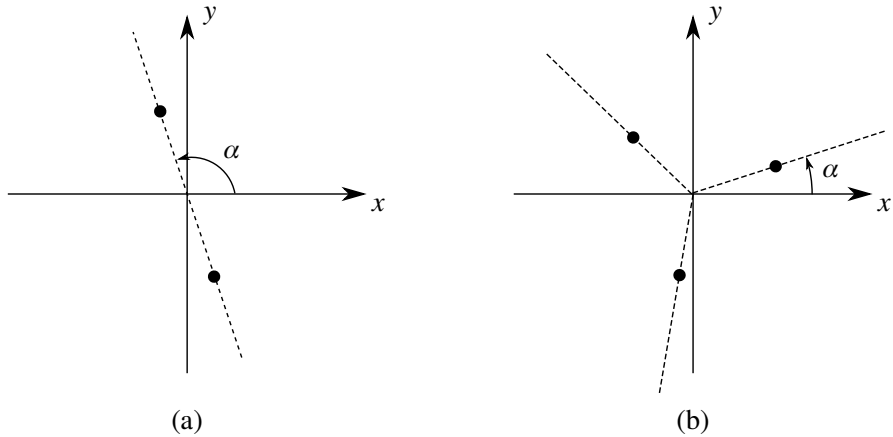}
\caption{Two simple composite scatterers (a) one with $2$-fold symmetry and (b) one with $3$-fold symmetry.}
\label{fig:scatterer_locations}
\end{center}
\end{figure}

The general formula \eqref{eq:general_point_scatterers} can be specialized to two subsets of scatterers, one with $P$-fold symmetry and one with $Q$-fold symmetry.  If we assume that the composite scatterer with $P$-fold symmetry consists of exactly $P$ points, then the signature reduces to a rather specific form
\begin{equation}
  \label{eq:one_scatterer_family}
  s_P(\theta,f;\alpha) = a \sum_{p=1}^P \frac{\exp\left(-\frac{2 \pi i f}{c} r \cos \left(\frac{2\pi p}{P} + \alpha -\theta\right)\right)}{\sqrt{\left(r \sin \left(\frac{2\pi p}{P} + \alpha-\theta\right)\right)^2 + \left(R + r \cos \left(\frac{2\pi p}{P} + \alpha-\theta\right)\right)^2}},
\end{equation}
in which all of the point scatterers have the same complex reflectivity $a=a_p$ and the same distance to target center $r = r_p$.  Because $\cos$ and $\sin$ are $2\pi$-periodic functions, and the sum consists of $P$ evenly spaced terms, it follows that $s(\theta + (2\pi/P),f) = s(\theta,f)$ for all look angles $\theta$ and frequencies $f$.

Example signatures of the two scatterers shown in Figure \ref{fig:scatterer_locations} are shown in Figure \ref{fig:scatterer_sigs}.  In the plots of the signatures, the look angle $\theta$ is shown in the vertical dimension and the frequency $f$ is shown on the horizontal dimension.  Notice that both of the signatures repeat vertically, with the $2$-fold symmetric target showing two copies and the $3$-fold symmetric target showing three copies.  It is therefore obvious that the symmetry class can be used to coarsely discriminate between targets, as the two signatures shown in Figure \ref{fig:scatterer_sigs} are easy to distinguish from each other.  Moreover, two targets with the same symmetry class can be discriminated by comparing their fundamental domains.

\begin{figure}
\begin{center}
\includegraphics[height=2in]{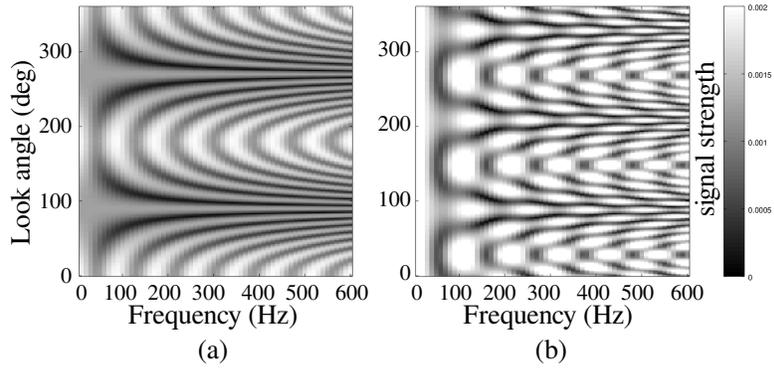}
\caption{Signatures of the composite scatterers shown in Figure \ref{fig:scatterer_locations}, (a) has $2$-fold symmetry and (b) has $3$-fold symmetry.}
\label{fig:scatterer_sigs}
\end{center}
\end{figure}

If the trajectory of the sensor is unevenly traversed so that the look angle does not vary linearly, then the periodicity visible in Figure \ref{fig:scatterer_sigs} is destroyed.  For instance, Figure \ref{fig:scatterer_sigs2} shows two possible signatures after a smooth trajectory distortion has been applied, though only one complete orbit around the target has been traversed.  Although Figure \ref{fig:scatterer_sigs}(b) and Figure \ref{fig:scatterer_sigs2}(a) are collections of echoes from the same composite scatterer and differ only by the application of a distortion $\phi$, they are quite different as functions.  On the other hand, no trajectory distortion of Figure \ref{fig:scatterer_sigs2}(b) can be applied to transform it into any of the others because its underlying symmetry class ($5$-fold symmetry) is different.

This article aims to make the intuitive ideas of removing trajectory distortions precise, and develops the proper notion of ``signatures equivalent up to trajectory distortions.''  The main idea is to realize that the signatures shown in Figure \ref{fig:scatterer_sigs} can act as representatives of signature equivalence classes, and to consider the properties arising from trajectory distortions.  Mathematically, decoupling trajectory effects from target effects requires that we consider the process of \emph{factoring functions via composition}.

\begin{figure}
\begin{center}
\includegraphics[height=2in]{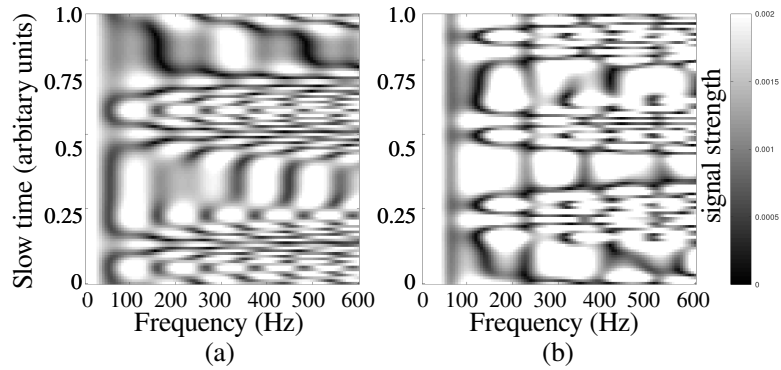}
\caption{Signatures of the composite scatterers with trajectory distortions over a single orbit of the target; after removing distortions (a) has $3$-fold symmetry and (b) has $5$-fold symmetry.}
\label{fig:scatterer_sigs2}
\end{center}
\end{figure}

As a preview, Theorem \ref{thm:circle_quasip} characterizes factorizations of functions with circular domains.  Intuitively, it says that one may distinguish the equivalence classes of signatures by their winding numbers.  In other words, even though the rows of Figure \ref{fig:scatterer_sigs2}(a) do not traverse look angles evenly, they repeat $3$ times.  This differs from Figure \ref{fig:scatterer_sigs2}(b), which makes $5$ repetitions, and so the signatures must be different.  On the other hand, the signature in Figure \ref{fig:scatterer_sigs}(b) also makes $3$ repetitions, so there are grounds for possible equivalence with Figure \ref{fig:scatterer_sigs2}(a).  Moreover, if we are considering only one frequency (column in Figures \ref{fig:scatterer_sigs} and \ref{fig:scatterer_sigs2}), then Corollary \ref{cor:circle_quasip} indicates that this is \emph{all} that can be determined.  We emphasize that none of these conclusions are surprising in the case that the trajectory is evenly traversed---the trajectory has been motion compensated to look angle---but we reiterate that ensuring trajectories are evenly traversed can be quite difficult in practice.

\subsection{Two composite scatterers: constant rank factorizations}

For a target formed as the superposition of two composite scatterers, Proposition \ref{prop:signal_knot} shows that the signature traces out a \emph{torus knot}, a closed path on the surface of a torus.  This remains true if there are trajectory distortions.  As a result, Theorem \ref{thm:circle_quasip} states that such signatures can be classified by the topological type of the corresponding torus knot.  Torus knots are classified by a pair of two integer winding numbers $(P,Q)$, which correspond to the symmetry classes of the composite scatterers.

Again using the look angle directly to fix ideas, we can construct the superposition of both subsets of scatterers, a $P$-fold symmetric subset and a $Q$-fold symmetric subset.  This is done by taking a sum\footnote{It will become apparent that Equation \ref{eq:two_scatterers} satisfies the hypotheses of Proposition \ref{prop:signal_knot}.  Given this fact, $s$ has a \emph{constant rank factorization} in which the \emph{phase function} is a torus knot.} of the two composite scatterers:
\begin{equation}
  \label{eq:two_scatterers}
  s(\theta,f;\alpha,\alpha') = s_P(\theta,f;\alpha) + s_Q(\theta,f;\alpha').
\end{equation}
It is sometimes helpful to consider the \emph{relative angle} $\beta = \alpha' - \alpha$. Figure \ref{fig:scatterer_relative_angle} shows the configuration of point scatterers if $P=2$, $Q=3$, and $\beta = 28^\circ$.

\begin{figure}
\begin{center}
\includegraphics[height=2.25in]{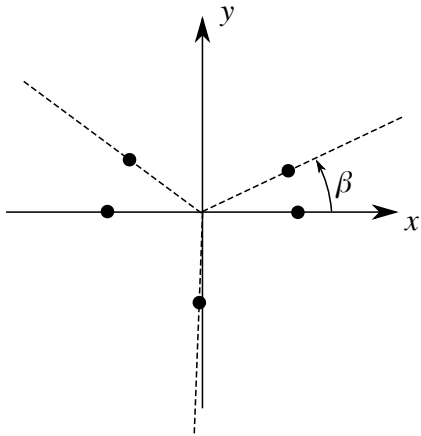}
\caption{Combined scatterer with $P=2$, $Q=3$, and relative angle $\beta = 28^\circ$ shown.}
\label{fig:scatterer_relative_angle}
\end{center}
\end{figure}

Since the $\theta=0$ axis was essentially arbitrary, it is sometimes helpful to take $\alpha=0$, resulting in the formula
\begin{eqnarray*}
  s(\theta,f;0,\beta) &=& a \sum_{p=1}^P \frac{\exp\left(-\frac{2 \pi i f}{c} r \cos \left(\frac{2\pi p}{P} -\theta\right)\right)}{\sqrt{\left(r \sin \left(\frac{2\pi p}{P} -\theta\right)\right)^2 + \left(R + r \cos \left(\frac{2\pi p}{P} -\theta\right)\right)^2}} \\&& + a' \sum_{q=1}^Q \frac{\exp\left(-\frac{2 \pi i f}{c} r' \cos \left(\frac{2\pi q}{Q} + \beta -\theta\right)\right)}{\sqrt{\left(r' \sin \left(\frac{2\pi q}{Q} + \beta-\theta\right)\right)^2 + \left(R + r' \cos \left(\frac{2\pi q}{Q} + \beta-\theta\right)\right)^2}}.
\end{eqnarray*}
For simplicity, let us take $a=a'$ and $r=r'$ in what follows in this section.  The resulting signature is shown in Figure \ref{fig:scatterer_combined_sig}.  

Notice that superposition of these particular two subsets (Figure \ref{fig:scatterer_relative_angle}) is not quite $2$-fold symmetric itself, because $28^\circ + 240^\circ = 268^\circ \not= 270^\circ$.  It would be $2$-fold symmetric if the relative angle were to be chosen as $\beta=30^\circ$, and there are other possible choices of relative angle resulting in a $2$-fold symmetric target.  The effect of this slight asymmetry is visible in Figure \ref{fig:scatterer_combined_sig}, where it makes the signature approximately repeat vertically.

\begin{figure}
\begin{center}
\includegraphics[height=2in]{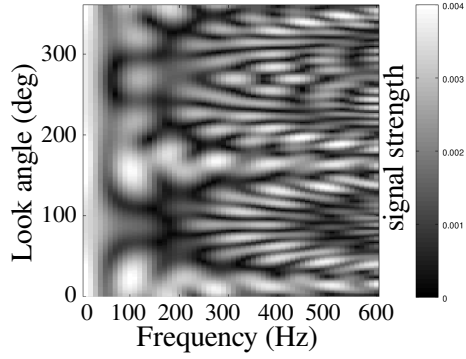}
\caption{A typical signature of the sum of the two scatters shown in Figure \ref{fig:scatterer_relative_angle}, where $P=2$, $Q=3$, and the relative angle is $\beta=28^\circ$.}
\label{fig:scatterer_combined_sig}
\end{center}
\end{figure}

After inspecting Equation \eqref{eq:one_scatterer_family}, it is clear that holding $\theta$ fixed and varying $\alpha$ results in the same information as varying $\theta$ while holding $\alpha$ fixed.  With two reference angles in Equation \eqref{eq:two_scatterers}, it is more convenient to hold the look angle $\theta$ fixed and vary the reference angles instead.  If we hold the look angle $\theta$ and frequency $f$ fixed, while varying the reference angles $\alpha$ and $\alpha'$ arbitrarily, we obtain a somewhat different perspective of the signature, shown for two different frequencies in Figure \ref{fig:toroidal_flat}.  

\begin{figure}
\begin{center}
\includegraphics[height=2in]{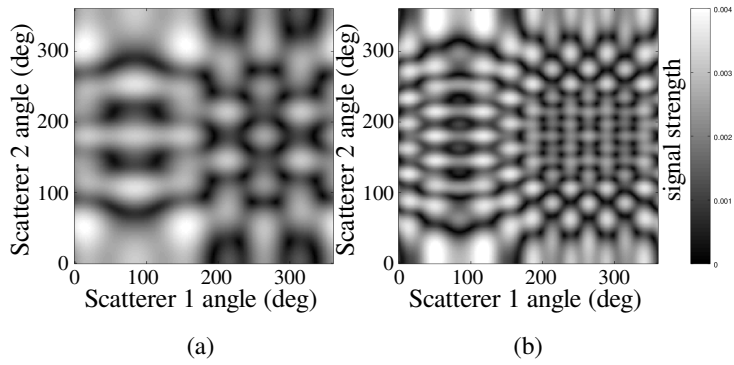}
\caption{Response as a function of scatterer angles at (a) $300$ Hz, (b) $600$ Hz.}
\label{fig:toroidal_flat}
\end{center}
\end{figure}

The benefit of this perspective is that all possible relative angles $\beta$ are represented.  The locus of points in Figure \ref{fig:toroidal_flat} with a fixed relative angle $\beta = \alpha' - \alpha$ is a line with slope of $1$.  The relative angle $\beta$ determines the intercept of the line.  For instance, if $\beta = 28^\circ$, Figure \ref{fig:crosscheck_slice_300} shows how the signature at $300$ Hz is obtained by following such a trajectory.

\begin{figure}
\begin{center}
\includegraphics[height=2in]{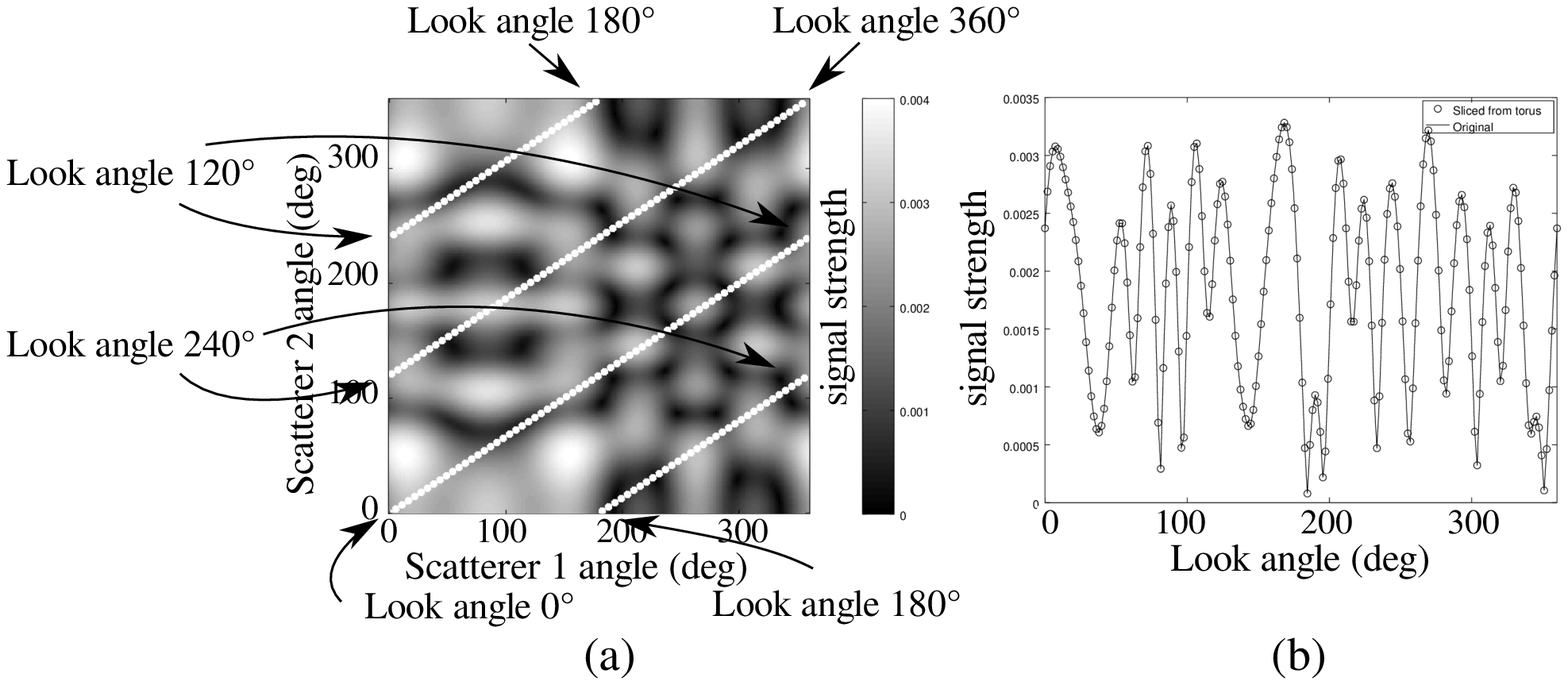}
\caption{Extracted single frequency response with a constant relative angle ($28^\circ$) between scatterers at $300$ Hz.}
\label{fig:crosscheck_slice_300}
\end{center}
\end{figure}

The plots shown in Figures \ref{fig:toroidal_flat}(a)--(b) and \ref{fig:crosscheck_slice_300}(a) are really on the surface of a torus in virtue of Proposition \ref{prop:signal_knot}, since both the horizontal and vertical axes measure angles.  The trajectory shown in Figure \ref{fig:crosscheck_slice_300}(a) for a relative angle $\beta=28^\circ$ eventually returns to its starting point (at the origin), after making two complete horizontal circuits and three complete vertical circuits in Figure \ref{fig:crosscheck_slice_300}(a).  The trajectory is therefore an example of a $(2,3)$ torus knot, as is immediately apparent when drawn on the surface of an embedded torus in Figure \ref{fig:ideal_knots}(a).  The signature is therefore best thought of as a function on the surface of the torus, since it too is periodic in both reference angles.  If we consider a different frequency, say $600$ Hz, then the torus knot trajectory is unchanged, but the underlying function on the surface of the torus changes, as shown in Figure \ref{fig:ideal_knots}(b).  Considering \emph{all} frequencies $f$ is theoretically (and practically) valuable, but is difficult to draw. 

\begin{figure}
\begin{center}
\includegraphics[height=2in]{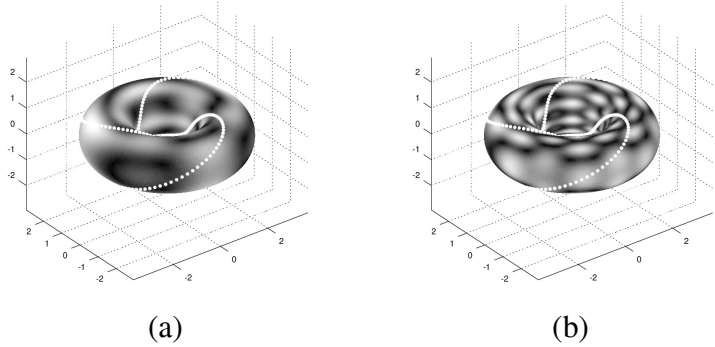}
\caption{Single frequency response with a constant relative angle ($28^\circ$) as a torus knot at (a) $300$ Hz, (b) $600$ Hz.}
\label{fig:ideal_knots}
\end{center}
\end{figure}

The functions on the surface of the torus in Figure \ref{fig:ideal_knots} allow us to go beyond merely classifying signatures by torus knot type.  If two targets have the same torus knot type, they can still be distinguished by differences between their corresponding functions on the torus.  If we switch to a different target, the torus functions will change, probably substantially, and will no longer assure the same kind of equivalence.  Therefore, the basic idea given two signatures is to first determine if they have the same torus knot type.  If not, they correspond to different targets.  If they do have the same torus knot type, then one would hope to simulate the functions on the torus and argue based on their specifics.

The mathematical invariants explored in this article characterize (theoretically) how close the two signatures are by parameterizing \emph{all} possible functions on the torus, and on all other manifolds besides.  While these invariants are \emph{emphatically not yet practically computable}, their use provides a theoretical framework for asserting that two targets are distinguishable \emph{regardless} of trajectory distortions.

The key insight is that distorting the sensor trajectory deflects the knot on the surface of the torus, but cannot change its torus knot type.  Effectively, because the knot is constrained to lie on the surface of the torus, its knot type is ``locked in place.''  This article defines \emph{constant rank factorizations} to model trajectory distortions in this case.  Moreover, the article defines an equivalence class that classifies these factorizations.  The article presents several tools to characterize the mathematical properties of constant rank factorizations (Theorems \ref{thm:coslice_equivalent} and \ref{thm:crse_coslice}), and presents a complete characterization (Theorem \ref{thm:circle_quasip}) that applies to $P$-fold symmetric targets.

\section{Categories of factorizations}
\label{sec:definitions}

We begin by generalizing the ideas inherent in Section \ref{sec:motivation}.  Given the nature of the sensing problem, the collected data are measured in a \emph{signal space} $N$, which is typically a submanifold of $\mathbb{C}^n$.  For instance, the signal space is $\mathbb{C}^n$ in Equations \eqref{eq:general_point_scatterers} \eqref{eq:one_scatterer_family}, and \eqref{eq:two_scatterers} if $n$ distinct frequencies are collected from each sensor location along the trajectory.

Similarly, the sensor's configuration (location, time, etc.) is captured by a manifold $M$.  We assume that a signature of our target is characterized by an idealized representative target model represented as a smooth map between two spaces $U: C \to N$, and that the sensor trajectory can be modeled by composing this idealized model with a function $\phi: M \to C$, which may contain distortions.  We cannot measure either of these functions, but instead are able to measure their composition $u = U \circ \phi$.  In order to make any theoretical headway, some assumptions must be placed on the functions $\phi$ or $U$.  This article relies on the idea that the trajectory does not ``come to a stop'' at any point.  This can be modeled formally by requiring the Jacobian matrix $d\phi$ to have a constant rank throughout $M$.

\begin{definition}
  \label{def:constant_rank}
  A \emph{constant rank factorization} of a smooth map $u: M \to N$ between two manifolds is a factorization $u = U \circ \phi$ in which $\phi$ is a smooth map of \emph{constant rank}\footnote{The Jacobian matrix $d_x\phi$ has the same rank at every point $x$ in $M$.}.  If we write $\phi : M \to C$ and $U : C \to N$ for a constant rank factorization, we call $C$ the \emph{phase space}, $\phi$ the \emph{phase map}, and $U$ the \emph{signature map}.  We will often write a constant rank factorization $u = U \circ \phi$ as an ordered pair $(\phi,U)$ when the function $u$ is understood from context.

  A \emph{morphism} $C: (\phi_1,U_1) \to (\phi_2,U_2)$ between two constant rank factorizations of a single map $u: M \to N$ consists of a commuting diagram\footnote{A diagram is said to \emph{commute} if the functions obtained by composing adjacent arrows depend only on their start and end points in the diagram.  Commutative diagrams are an efficient way to express a collection of functional equations, and make frequent appearances in this article.}
  \begin{equation*}
    \xymatrix{
      M \ar[rr]^u \ar[dr]^{\phi_1} \ar@/_1pc/[ddr]_{\phi_2} && N \\
      &C_1 \ar[d]_-{c} \ar[ur]^{U_1}&\\
      &C_2 \ar@/_1pc/[uur]_{U_2}&\\
      }
  \end{equation*}
  where $c: C_1 \to C_2$ is a smooth map, which we call the \emph{component map}.
  Because the diagram commutes, this means that $\phi_2 = c \circ \phi_1$ and $U_1=U_2 \circ c$.
\end{definition}

In the CSAS example in Section \ref{sec:motivation}, the path of the sensor defines the manifold $M = \mathbb{R}$.  The signal space is $N = \mathbb{C}^n$ where $n$ is the number of frequency samples collected.  Finally, the phase space $C=S^1$ is the look angle.  As a result, Figures \ref{fig:scatterer_sigs} and \ref{fig:scatterer_combined_sig} show examples of three signature maps ($U$ functions).  In contrast, Figure \ref{fig:scatterer_sigs2} shows two examples of $u= U \circ \phi$ functions.

\begin{lemma}
  The class of constant rank factorizations of a smooth map $u: M \to N$ forms a category $\cat{Const}(u)$, where the composition of morphisms of constant rank factorizations consists of composition of maps between the phase spaces.
\end{lemma}
\begin{proof}
  We take each constant rank factorization $u = U \circ \phi$ as an object in $\cat{Const}(u)$.

  Using the identity function for the component map, the morphism $\id_{(\phi,U)} : (\phi,U) \to (\phi,U)$ given by the diagram
\begin{equation*}
    \xymatrix{
      M \ar[rr]^u \ar[dr]^{\phi} \ar@/_1pc/[ddr]_{\phi} && N \\
      &C \ar[d]_-{\id} \ar[ur]^{U}&\\
      &C \ar@/_1pc/[uur]_{U}&\\
      }
\end{equation*}
does not change any other morphism when it is composed on the left or the right.  For instance, abusing notation slightly, if $m: (\phi,U) \to (\phi',U')$ is another morphism whose component map is $m$, then we have the larger diagram
\begin{equation*}
    \xymatrix{
      M \ar[rr]^u \ar[dr]^{\phi} \ar@/_1pc/[ddr]_{\phi} \ar@/_2pc/[dddr]_{\phi'}&& N \\
      &C \ar[d]_-{\id_C} \ar[ur]^{U}&\\
      &C \ar@/_1pc/[uur]_{U} \ar[d]_-{m}&\\
      &D \ar@/_2pc/[uuur]_{U'}&\\
      }
\end{equation*}
that defines the composition $m \circ \id_{(\phi,U)}$.  Evidently the component map of this composition is $m = m \circ \id_C$, which asserts that on the level of $\cat{Const}(u)$ morphisms, we have that $m=m \circ \id_{(\phi,U)}$.  Composition on the left by $\id_{(\phi,U)}$ follows similarly.

Suppose that we have three morphisms $m_1: (\phi_1,U_1) \to (\phi_2,U_2)$, $m_2: (\phi_2,U_2) \to (\phi_3,U_3)$, and $m_3: (\phi_3,U_3) \to (\phi_4,U_4)$.  Associativity of composition of these morphisms $(m_1 \circ m_2) \circ m_3 = m_1 \circ (m_2 \circ m_3)$ follows because the same equation holds for their corresponding component maps. 
\end{proof}

One general source of constant rank factorizations is \emph{quasiperiodic factorizations}, which were defined in several earlier papers.

\begin{definition}\cite{Robinson_qplpf,Robinson_SampTA_2015}
  A \emph{quasiperiodic factorization of a smooth map $u: M \to N$} consists of a factorization $u = U \circ \phi$ in which $\phi$ is a smooth submersion\footnote{A \emph{submersion} is a constant rank map in which the Jacobian matrix is surjective as a linear map.}.  The category of quasiperiodic factorizations $\cat{QuasiP}(u)$ is defined analogously to $\cat{Const}(u)$, with objects consisting of factorizations and morphisms given by diagrams of exactly the same form as in Definition \ref{def:constant_rank}.
\end{definition}

\begin{proposition}
  \label{prop:quasi_subcat}
  Every quasiperiodic factorization is also a constant rank factorization.  Thus the category of quasiperiodic factorizations $\cat{QuasiP}(u)$ is a subcategory of the category of constant rank factorizations $\cat{Const}(u)$.
\end{proposition}
\begin{proof}
  Observe that a submersion is necessarily of constant rank.
\end{proof}

Both categories $\cat{QuasiP}(u)$ and $\cat{Const}(u)$ share an initial object.

\begin{proposition}
  \label{prop:initial}
  The factorization of a smooth map $u: M \to N$ given by $u = u \circ \id_M$ is an initial object for $\cat{Const}(u)$ and an initial object for $\cat{QuasiP}(u)$.  
\end{proposition}
\begin{proof}
  To establish this, suppose that $u = U \circ \phi$ is another constant rank (or quasiperiodic) factorization.  The only way to make the following diagram commute
\begin{equation*}
  \xymatrix{
    M \ar[rr]^u \ar[dr]^{\id_M} \ar@/_1pc/[ddr]_{\phi} & & N \\
    & M \ar[ur]^u \ar[d]_{c}& \\
    & C \ar@/_1pc/[uur]_{U}\\
    }
\end{equation*}
for some manifold $C$ is to assert that $c=\phi$.  Conversely, with $c=\phi$, this diagram always commutes!
\end{proof}

Section \ref{sec:motivation} provides anecdotal evidence that certain sonar collections can be described by torus knots.  This is a general situation that arises whenever two periodic smooth functions are superposed.

\begin{proposition}
  \label{prop:signal_knot}
  If $u: S^1 \to \mathbb{C}^D$ is the sum
  \begin{equation*}
    u(x) = u_1(x) + u_2(x)
  \end{equation*}
  of two non-constant smooth functions, a $(2\pi / m)$-periodic function $u_1$, and a $(2\pi / n)$-periodic function $u_2$ for integers $m$ and $n$, then $u$ has a constant rank factorization in which the phase function is a $(m,n)$ torus knot.
\end{proposition}

Notice in particular that a torus $S^1 \times S^1$ is of larger dimension ($2$) than that of the circle $S^1$.  Therefore, the constant rank factorization guaranteed by this Proposition is \emph{not} a quasiperiodic factorization.  On the other hand, projecting to either of the two factors $S^1 \times S^1 \to S^1$ yields a quasiperiodic factorization.

\begin{proof}
  Let $\phi_1, \phi_2 : S^1 \to S^1$ be linear maps of degree $m$ and $n$, respectively, so that
  \begin{equation*}
    \phi_1(x) := [mx]_{2\pi}, \text{ and } \phi_2(x) := [nx]_{2\pi}.
  \end{equation*}
  According to \cite[Thm. 8]{Robinson_SampTA_2015}, this choice results in the universal quasiperiodic factorizations of $u_1$ and $u_2$ provided $m$ and $n$ are both minimal.
  If we define $\phi : S^1 \to (S^1 \times S^1)$ by
  \begin{equation*}
    \phi(x) := (\phi_1(x), \phi_2(x)) = ([mx]_{2\pi},[nx]_{2\pi})
  \end{equation*}
  then, by hypothesis
  \begin{equation*}
    \xymatrix{
      S^1 \ar[r]^u \ar[d]_{\phi} & \mathbb{C}^D \\
      S^1 \times S^1 \ar[ur]_{U}
      }
  \end{equation*}
  commutes if we take
  \begin{equation*}
    U(y,z):= u_1(y) + u_2(z).
  \end{equation*}
  By construction, $\phi$ is of constant rank, and is also a torus knot of type $(m,n)$.
\end{proof}


The categories $\cat{Const}(u)$ and $\cat{QuasiP}(u)$ are diffeomorphism invariants.

\begin{theorem}
  \label{thm:crse_diffeo}
  Suppose that $u_1,u_2: M \to N$ are smooth maps and $M$ is a connected manifold.  If $u_2 = u_1 \circ f$ where $f: M \to M$ is a diffeomorphism, then
  \begin{enumerate}
  \item $\cat{Const}(u_1)$ and $\cat{Const}(u_2)$ are isomorphic categories, and
  \item $\cat{QuasiP}(u_1)$ and $\cat{QuasiP}(u_2)$ are isomorphic categories.
  \end{enumerate}
\end{theorem}
\begin{proof}
  Both statements follow from reasoning about the following commutative diagram
  \begin{equation*}
    \xymatrix{
      M \ar[d]_\phi \ar[dr]^{u_1} & M \ar[l]_f \ar[d]^{u_2}\\
      C \ar[r]_U & N
      }
  \end{equation*}
  for an arbitrary constant rank (or quasiperiodic factorization) $u_1 = U \circ \phi$.  Evidently this implies that $u_2 = U \circ (\phi \circ f)$ is a factorization of the same type, since diffeomorphisms are of locally constant rank and we assumed that $M$ is connected.  This provides for a bijection on the objects of the two categories under discussion.  Since this also means that the factorizations of $u_1$ and $u_2$ share the space $C$, there is no further transformation required on morphisms to establish the isomorphism.
\end{proof}

A given smooth map $u$ often has many quasiperiodic factorizations, though there is a ``simplest'' one given by the final object of $\cat{QuasiP}(u)$.  That these final objects always exist is proven in \cite[Thm. 5]{Robinson_SampTA_2015}.  The same proof extends to $\cat{Const}(u)$ without incident.

\begin{proposition}
  \label{prop:universal_constant_rank}
  The category of constant rank factorizations $\cat{Const}(u)$ of a smooth map $u: M \to N$ has final objects.
\end{proposition}
\begin{proof}
  The construction given in \cite[Thm. 5]{Robinson_SampTA_2015} of a universal quasiperiodic factorization goes through \emph{mutatis mutandis} for constant rank factorizations.
\end{proof}

At this point it likely seems that constant rank factorizations and quasiperiodic factorizations are rather similar.  The key difference between these two concepts is crystallized by the idea of an embedded torus knot $u: S^1 \to (S^1)^n \subset \mathbb{R}^{n+1}$.  Every quasiperiodic factorization of $u$ will necessarily have a $1$-dimensional phase space due to the requirement that the phase map be a surjective submersion.  (Moreover, see Lemma \ref{lem:loop_phase_space} proven a little later in this article.)  In contrast, factorization of $u$ through a torus $(S^1)^n$ can be of constant rank.  In this way, torus knots are more naturally studied through the lens of constant rank factorizations.  The next example emphasizes this difference a bit more starkly by exhibiting a torus knot that has no nontrivial quasiperiodic factorizations.  The key idea is to violate the periodicity hypotheses of Proposition \ref{prop:signal_knot}.

\begin{example}
  \label{eg:infinite_knot}
  The function $u: \mathbb{R} \to \mathbb{R}$ given by
  \begin{equation*}
    u(x) := \sin x + \sin(\pi x)
  \end{equation*}
  has a constant rank factorization
  \begin{equation*}
    \xymatrix{
      \mathbb{R} \ar[r]^u \ar[d]_{\phi} & \mathbb{R} \\
      S^1 \times S^1 \ar[ur]_U& \\
    }
  \end{equation*}
  where
  \begin{equation}
    \label{eq:phi_sinx_sinpix}
    \phi(x) := \left([x]_{2\pi},[\pi x]_{2\pi}\right),
  \end{equation}
  and
  \begin{equation*}
    U(y,z) := \sin y + \sin z,
  \end{equation*}
  since the derivative of $\phi$ is a constant, hence of constant rank.  Since the dimension of $\mathbb{R}$ is $1$, it follows that $\phi$ cannot be a submersion $\phi: \mathbb{R} \to (S^1 \times S^1)$ since the codomain is of dimension $2$.

  On the other hand, $u$ has no quasiperiodic factorization with $S^1$ as the phase space, and so only has the trivial quasiperiodic factorization.  To see this, recall the trigonometric identity
  \begin{eqnarray*}
    \sin(x+y)\cos(x-y) &=& \frac{1}{4i}\left(e^{i(x+y)}-e^{-i(x+y)}\right)\left(e^{i(x-y)}+e^{-i(x-y)}\right)\\
    &=& \frac{1}{4i}\left(e^{2ix} - e^{-2ix}+e^{2iy}-e^{-2iy}\right)\\
    &=& \frac{1}{2}\sin (2x) + \frac{1}{2}\sin(2y).
  \end{eqnarray*}
  Thus
  \begin{equation*}
    u(x) = 2 \sin\left(\left(\frac{1+\pi}{2}\right)x\right)\cos\left(\left(\frac{1-\pi}{2}\right)x\right).
  \end{equation*}
  The supremum of this function is certainly not more than $2$ and its infimum is certainly not less than $-2$.  It is a commonly-recognized fact that the image of $\phi:\mathbb{R}\to \left(S^1 \times S^1\right)$ given as Equation \eqref{eq:phi_sinx_sinpix} is dense in $S^1 \times S^1$, so this means that the supremum of $u$ is therefore equal to $2$ and the infimum is equal to $-2$.

  Now suppose that there was a quasiperiodic factorization $u = V \circ \psi$ with phase space $S^1$, so that $\psi: \mathbb{R} \to S^1$ and $V : S^1 \to \mathbb{R}$.  Since $S^1$ is compact, $V$ must attain its maximum value, which must be $2$, and likewise $V$ must attain its minimum value, which must be $-2$.  Therefore, at the maximum, $x$ must satisfy the equation
  \begin{equation*}
    2 = 2 \sin\left(\left(\frac{1+\pi}{2}\right)x\right)\cos\left(\left(\frac{1-\pi}{2}\right)x\right).
  \end{equation*}
  This means that
  \begin{equation*}
    \left(\frac{1+\pi}{2}\right)x = \left(\frac{1}{2}+2n\right)\pi,\text{ and }\left(\frac{1-\pi}{2}\right)x = 2 m \pi
  \end{equation*}
  for some integers $m$ and $n$.  Therefore,
  \begin{equation*}
    x = \frac{4 m \pi}{1-\pi}
  \end{equation*}
  so that
  \begin{equation*}
    \frac{1+4n}{1+\pi} = \frac{4m}{1-\pi}.
  \end{equation*}
  Solving for $m$, we have that
  \begin{equation*}
    m = \frac{1}{4}\left(\frac{1-\pi}{1+\pi}\right)(1+4n),
  \end{equation*}
  which is a contradiction since that quantity is not rational!  
\end{example}

Taken together, Propositions \ref{prop:quasi_subcat} and \ref{prop:universal_constant_rank} imply that any final object $u = U \circ \phi$ of $\cat{QuasiP}(u)$ is an object of $\cat{Const}(u)$.
One therefore might speculate that $u$ would also be a final object of $\cat{Const}(u)$ as well.
The map in Example \ref{eg:infinite_knot} is a counterexample to this, however!
Specifically, take an object of $\cat{Const}(u)$ with a compact phase space (namely $S^1 \times S^1$).
There can be no $\cat{Const}(u)$ morphism from this object to the single isomorphism class of $\cat{QuasiP}(u)$ since the phase space in that case is not compact; this is precluded by the constant rank assumption.

\subsection{Functoriality}
\label{sec:functoriality}

The categories $\cat{Const}(u)$ and $\cat{QuasiP}(u)$ are natural in the sense that they are transformed functorially by pre- and post-composition of $u$ with other smooth functions.  The following Propositions are modifications of Theorem \ref{thm:crse_diffeo}.

\begin{proposition}
  \label{prop:left_functor}
  Suppose that $u: M \to N$ and $f: N \to N'$ are smooth maps.  The map $f$ induces a covariant functor $\cat{QuasiP}(u) \to \cat{QuasiP}(f \circ u)$ and a covariant functor $\cat{Const}(u) \to \cat{Const}(f \circ u)$.
\end{proposition}
\begin{proof}
  Suppose $u = U \circ\phi$ is a quasiperiodic or a constant rank factorization in which $C$ is the phase space.  Then the following diagram commutes
  \begin{equation*}
    \xymatrix{
    M \ar[rr]^u \ar[dr]_\phi && N \ar[r]^f & N' \\
    &C \ar[ur]_U \ar@/_1pc/[urr]_{f \circ U}
    }
  \end{equation*}
  which means that $(f \circ u) = (f \circ U) \circ \phi$ is a factorization of the same type as for $u$.  This transforms the objects of the categories $\cat{QuasiP}(u)$ or $\cat{Const}(u)$.

  Morphisms are transformed in much the same way as objects, by composing $f$ on the left, as can be seen from the commutative diagram
  \begin{equation*}
    \xymatrix{
    M \ar[rr]^u \ar[dr]_\phi \ar@/_1pc/[ddr]_{\phi'} && N \ar[r]^f & N' \\
    &C \ar[ur]_U \ar@/_1pc/[urr]_{f \circ U} \ar[d]_c\\
    &C' \ar@/_1pc/[uur]_(0.4){U'}|(0.625)\hole \ar@/_2pc/[uurr]_{f \circ U'}\\
    }
  \end{equation*}

  Composition of morphisms is completely unchanged by this recipe, as it consists of composition of the maps on the phase space.
\end{proof}

\begin{proposition}
  \label{prop:right_functor}
  Suppose that $u: M \to N$ is a smooth map.  If $g: M' \to M$ is a surjective submersion, then $g$ induces a covariant functor $\cat{QuasiP}(u) \to \cat{QuasiP}(u \circ g)$.  Similarly, if $g: M' \to M$ is a constant rank map, then $g$ induces a covariant functor $\cat{Const}(u) \to \cat{Const}(u \circ g)$.
\end{proposition}
\begin{proof}
  Suppose $u = U \circ\phi$ is a quasiperiodic factorization in which $C$ is the phase space.  Because $g$ is assumed to be a surjective submersion, $(\phi \circ g)$ is also a surjective submersion.  Therefore, the following diagram commutes
  \begin{equation*}
    \xymatrix{
    M' \ar[r]^g \ar@/_1pc/[drr]_{\phi \circ g} &M \ar[rr]^u \ar[dr]_\phi && N \\
    &&C \ar[ur]_U
    }
  \end{equation*}
  which means that $(u \circ g) = U \circ (\phi \circ g)$ is a quasiperiodic factorization.   Similarly, if $u= U \circ \phi$ is a constant rank factorization, then under the assumption that $g$ is a constant rank function, then so is $(\phi \circ g)$.  This transforms the objects of the categories $\cat{QuasiP}(u)$ or $\cat{Const}(u)$.  Morphisms and their composition follow along \emph{mutatis mutandis} as in the proof of Proposition \ref{prop:left_functor}.
\end{proof}

\subsection{Functions with circular domains}
\label{sec:circular}

At present, the most complete characterization of $\cat{QuasiP}(u)$ that is known is Theorem \ref{thm:circle_quasip}, which is proven in this section.  Theorem \ref{thm:circle_quasip} applies when the domain $M$ is the circle $S^1$.  This characterization relies on the observation that $u$ becomes a loop in $N$, and therefore has a representative $[u]$ in the fundamental group $\pi_1(N)$.

\begin{example}
  Let us consider the identity map $u=\id_{S^1}: S^1 \to S^1$ and its category $\cat{Const}(\id_{S^1})$ of constant rank factorizations.  First of all, its factorization $u = u \circ u$ is the final object of $\cat{Const}(u)$.  While this seems a little trivial, suppose that we had some other constant rank factorization $u = U \circ \phi$ with $S^1$ as its phase space.  If $\phi$ is not injective, this means that $U \circ \phi$ is not injective, which contradicts its factorization since $u$ is injective.  On the other hand, suppose that $\phi: S^1 \to S^1$ was not surjective.  Let $x \in S^1$ be outside the image of $\phi$.  Since $\phi$ is continuous, there is an open neighborhood of $x$ outside the image of $\phi$, which amounts to saying that the image of $\phi$ is homeomorphic to a closed interval.  Such a map $\phi$ must therefore have critical points, and therefore cannot have constant rank.  Therefore, $\phi$ must be bijective to participate in a constant rank factorization of $u = U \circ \phi$.  Since $\phi$ is of constant rank, it is therefore a local diffeomorphism; bijectivity assures that it is a diffeomorphism.  

  Also, $u$ cannot have $\mathbb{R}$ as the phase space of any of its quasiperiodic factorizations.  Since $S^1$ is compact, this would imply that the phase map $S^1 \to \mathbb{R}$ has at least one local maximum, which is a critical point.  This violates the constant rank assumption.

  On the other hand, with constant rank factorizations, there are many other possible phase spaces.
  Perhaps the easiest such is to take the cylinder $S^1 \times \mathbb{R}$ as the phase space, with $\phi(x) = (x, 0)$ and $U(x,y) = x$.  It is clear that $\phi$ is of constant rank (it is $1$), and since the second factor in the cylinder is ignored, this assures that $u= U \circ \phi$.  This factorization is not isomorphic to the trivial one in $\cat{Const}(u)$, because that would imply the existence of a smooth, bijective map $c: S^1 \to (S^1 \times \mathbb{R})$ such that the diagram below commutes
  \begin{equation*}
    \xymatrix{
      S^1 \ar[rr]^{\id} \ar[dr]^{\id} \ar@/_1pc/[ddr]_{\phi} && S^1 \\
      &S^1 \ar@/_/[d]_-{c} \ar[ur]^{\id}&\\
      &(S^1 \times \mathbb{R}) \ar@/_/[u]_{c^{-1}} \ar@/_1pc/[uur]_{U}&\\
      }
  \end{equation*}
Such is evidently impossible on several grounds, the most damning of which is the classical invariance of dimension because $S^1$ and $S^1 \times \mathbb{R}$ are of differing dimension.
\end{example}

The above example illuminates a general principle about quasiperiodic factorizations whose domain is a circle.

\begin{lemma}
  \label{lem:loop_phase_space}
  Every quasiperiodic factorization of a smooth map $u: S^1 \to N$ has either $S^1$ or the single-point space as its phase space.
\end{lemma}

This does not hold for constant rank factorizations since the phase map need not be surjective in that case.

\begin{proof}
  If $u=U\circ \phi$ is a quasiperiodic factorization, this means that $\phi$ is a surjective submersion.  Therefore, the only options for the phase space are $0$ or $1$ dimensional.  Additionally, because $\phi$ is continuous and $S^1$ is connected and compact, the phase space must also be connected and compact.  One option is evidently the single-point space.  Setting this aside, there are only two compact $1$-dimensional manifolds (with or without boundary), namely a closed interval or $S^1$.  Any smooth map $\phi$ from $S^1$ to the closed interval must have at least one critical point, since the interval can be totally ordered and thus a maximum value is attained via compactness.  The presence of critical points violates the requirement that $\phi$ be a submersion, which leaves $S^1$ as the only possible $1$-dimensional phase space.
\end{proof}

Recall that since $H_1(S^1) \cong \mathbb{Z}$, every continuous map $u : S^1 \to S^1$ induces a group homomorphism $u_* : H_1(S^1) \to H_1(S^1)$ of the form
\begin{equation*}
  u_*(n) = k n
\end{equation*}
for some integer $k$.  We call $k$ the \emph{degree $deg(u)$} of $u$.

\begin{lemma}(Standard; see \cite[Sec. 2.2]{Hatcher_2002}, for instance)
  \label{lem:degree_factoring}
  If $u, u_1, u_2 : S^1 \to S^1$ are smooth maps such that $u = u_2 \circ u_1$, then $deg(u) = deg(u_2) deg(u_1)$.
\end{lemma}
\begin{proof}
 Consider the map $u_*: H_1(S^1) \to H_1(S^1)$ induced by $u$ on $H_1(S^1)$.  We have that $u_*([z]) = deg(u) [z]$ where $[z]$ is the generator of $H_1(S^1) \cong \mathbb{Z}$.  But since homology is functorial, we have that $(u_2 \circ u_1)_*([z]) = (u_2)_* (u_1)_* ([z]) = deg(u_2) deg(u_1) [z]$.
\end{proof}

Another way to see that Lemma \ref{lem:degree_factoring} is true is to recall that the prototypical map of degree $n$ is $z^n : \mathbb{C} \to \mathbb{C}$ in the complex plane, and composition of this with $z^m$ yields $(z^n)^m = z^{mn}$.

Taking Lemma \ref{lem:loop_phase_space} to its next logical step, if $U \circ \phi$ is a quasiperiodic factorization of $u: S^1 \to N$, then $\phi: S^1\to S^1$ has a well defined degree $deg(\phi)$.
A useful consequence of Lemma \ref{lem:degree_factoring} is that degrees characterize isomorphism classes of $\cat{QuasiP}(u)$ up to sign changes in some cases.

\begin{lemma}
  \label{lem:loop_phase_function}
  Two quasiperiodic factorizations $U_1 \circ \phi_1$ and $U_2 \circ \phi_2$ of a smooth map $u: S^1 \to N$ are isomorphic in $\cat{QuasiP}(u)$ if and only if $|deg(\phi_1)| = |deg(\phi_2)|$.
\end{lemma}
\begin{proof}
Suppose that we have two such factorizations $u=U_1 \circ \phi_1 = U_2 \circ \phi_2$ which correspond to isomorphic objects of $\cat{QuasiP}(u)$.  This means that we have a continuous map $f$ such that the diagram commutes
\begin{equation}
  \label{eq:circle_quasip}
  \xymatrix{
    &S^1\ar@/^/[dd]^f \ar[dr]^{U_1}&\\
    S^1 \ar[ur]^{\phi_1} \ar[dr]_{\phi_2} && N \\
    &S^1\ar@/^/[uu]^{f^{-1}} \ar[ur]_{U_2}&\\
    }
\end{equation}
In particular, this implies that $\phi_2 = f \circ \phi_1$ and $\phi_1 = f^{-1} \circ \phi_2$.  In terms of degrees, this means that $deg(\phi_2) = deg(f) deg( \phi_1)$ and $deg(\phi_1) = deg(f^{-1}) deg( \phi_2)$.  Since all of these degrees must be nonzero integers, we have to conclude that $deg(f) = deg(f^{-1}) = \pm 1$, namely that $|deg(\phi_1)| = |deg(\phi_2)|$.  

Conversely, suppose that we have two quasiperiodic factorizations $u=U_1 \circ \phi_1 = U_2 \circ \phi_2$ with $|deg(\phi_1)| = |deg(\phi_2)|$.  Does this imply the existence of a homeomorphism $f$ such that the diagram in \eqref{eq:circle_quasip} commutes?  Yes, the covering space classification theorem \cite[Thm 1.38]{Hatcher_2002} yields an explicit construction of such a map $f$.  Thus, the two quasiperiodic factorizations are isomorphic in $\cat{QuasiP}(u)$.
\end{proof}

The above Lemmas lead to the following characterization of isomorphism classes of objects in $\cat{QuasiP}(u)$, at least when the domain of $u$ is a circle.

\begin{theorem}
  \label{thm:circle_quasip}
  If $u: S^1 \to N$ is a smooth map, then the isomorphism classes of $\cat{QuasiP}(u)$ are in bijective correspondence with cyclic subgroups of $\pi_1(N)$ that contain $[u]$.
\end{theorem}

\begin{proof}
  First of all, if $u$ is constant, then there is nothing to prove.  Let us therefore assume that $u$ is not constant for the remainder of the argument.
  
  Suppose that $u=U \circ \phi$ is a quasiperiodic factorization.  We will use this to generate a cyclic subgroup of $\pi_1(N)$ that contains $[u]$.
  Because of Lemma \ref{lem:loop_phase_space}, $U: S^1 \to N$ therefore corresponds to an element $[U] \in \pi_1(N)$.
  Let $T: Ob(\cat{QuasiP}(u)) \to 2^{\pi_1(N)}$ be given by
  \begin{equation*}
    T(U \circ \phi) := \{[U]^k : k \in \mathbb{Z}\},
  \end{equation*}
  which is evidently a cyclic subgroup of $\pi_1(N)$.
  Moreover, Lemma \ref{lem:loop_phase_function} indicates that $[u]$ is homotopic to $[U]^{deg(\phi)}$, and therefore $T(U \circ \phi)$ contains $[u]$ as required.

  Suppose that $H \subseteq \pi_1(N)$ is a cyclic subgroup containing $[u]$.
  That means that $[u]$ is homotopic to $[U]^k$ for some $k$ and some $U: S^1 \to N$.  We need to find a function $\phi: S^1 \to S^1$ such that $u = U \circ \phi$.  By the covering space classification theorem \cite[Thm 1.38]{Hatcher_2002}, one can construct such a $\phi$ that additionally satisfies $k=deg(\phi)$.
  
  If two quasiperiodic factorizations of $u=U_1\circ \phi_1=U_2\circ\phi_2$ are $\cat{QuasiP}(u)$-isomorphic, then they generate the same cyclic subgroup of $\pi_1(N)$.
  To see this, note that $|deg(\phi_1)| = |deg(\phi_2)|$ by Lemma \ref{lem:loop_phase_function}, and therefore $[u]$ is homotopic to both $[U_1]^{deg(\phi_1)}$ and $[U_2]^{deg(\phi_2)}$.  Since both are cyclic subgroups, this means that $[U_1]^k$ and $[U_2]^{\pm k}$ are homotopic, so that $T(U_1 \circ \phi_1)$ and $T(U_2 \circ \phi_2)$ are the same subgroup.

  Conversely, if two quasiperiodic factorizations of $u=U_1\circ \phi_1=U_2\circ\phi_2$ generate the same cyclic subgroup of $\pi_1(N)$ via $T$, this means that $deg(\phi_1) = \pm deg(\phi_2)$ and that $[U_1]$ is homotopic to $[U_2]$.  As a result of Lemma \ref{lem:loop_phase_function}, these two factorizations are isomorphic in $\cat{QuasiP}(u)$.
\end{proof}

\begin{corollary}
  \label{cor:circle_quasip}
  Suppose that $u : S^1 \to S^1$ is a smooth map of degree $n$.  Then the isomorphism classes of $\cat{QuasiP}(u)$ are in bijective correspondence with the set of factors of $|n|$.
\end{corollary}

\begin{proof}
First of all, isomorphism classes of $\cat{QuasiP}(u)$ are characterized by Lemma \ref{lem:loop_phase_function}.  Secondly, by Lemma \ref{lem:degree_factoring}, any quasiperiodic factorization into $u = U \circ \phi$ will yield $n = deg(u) = deg(U) deg(\phi)$.  
\end{proof}

\begin{corollary}
  \label{cor:circle_homotopy_quasip}
  If $u_1, u_2 : S^1 \to N$ are homotopic maps, then their categories $\cat{QuasiP}(u_1)$ and $\cat{QuasiP}(u_2)$ are equivalent.
\end{corollary}

\section{Comparing functions via common factors}
\label{sec:comparing}

Being apparently topological in nature, one might imagine that if $u_1, u_2 : M \to N$ are homotopic maps, then their categories $\cat{QuasiP}(u_1)$ and $\cat{QuasiP}(u_2)$ should be equivalent.
This turns out to be false, the results of Section \ref{sec:circular} (and Corollary \ref{cor:circle_homotopy_quasip} in particular) notwithstanding.  This section presents two counterexamples to the equivalence of homotopy and the factorization categories.  Taken together, these two counterexamples establish that $\cat{QuasiP}$ is a topological invariant distinct from homotopy.

Given that the constructions of $\cat{QuasiP}$ and $\cat{Const}$ were motivated by the needs of sonar signatures, this means that the topological information inherent in sonar signatures is not entirely captured by homotopy equivalence classes for spaces of echoes.  On the other hand, homotopy equivalence classes and equivalence classes of $\cat{QuasiP}$ and $\cat{Const}$ categories \emph{are} related, as will be explored in subsequent sections of this article.

\begin{example}
Consider $u_1,u_2: [0,1] \to [0,1]$ in which $u_1 = 0$ and $u_2 = \id$.  These two maps are evidently homotopic since $h:[0,1]\times[0,1] \to [0,1]$ given by
\begin{equation*}
  h(x,t) := t x
\end{equation*}
is a smooth homotopy between them.
Since $u_2$ has rank $1$, this means that any quasiperiodic factorization of it $u_2 = U \circ \phi$ must be a factorization into two rank $1$ functions.  This effectively means that there is precisely one choice of phase space up to diffeomorphism, and consequently $\cat{QuasiP}(u_2)$ contains exactly one isomorphism class.
On the other hand, $u_1$ need not factor into two rank $1$ functions; indeed the rank of $\phi$ may be $0$ or $1$.  Therefore, $\cat{QuasiP}(u_1)$ has two isomorphism classes.
\end{example}

Having equivalent categories of factorizations does not imply that two maps are homotopic, either.

\begin{example}
    Consider the identity map $u: S^1 \to S^1$ and the antipodal map $(-u) : S^1 \to S^1$.  The maps $u$ and $(-u)$ have equivalent categories of quasiperiodic factorizations, since both are their own universal quasiperiodic factorizations (a consequence of \cite[Thm. 8]{Robinson_SampTA_2015}), but they are not homotopic maps.
\end{example}

The net effect of these two examples is that $\cat{QuasiP}(u)$ is a topological invariant of a smooth map $u: M\to N$ that is distinct from the homotopy class of $u$.  There is still a somewhat more subtle relationship between homotopies and constant rank factorization categories.

\subsection{Signature equivalence}

If we interpret a constant rank factorization of $u = U \circ \phi$ as a model of a sonar collection, it is natural to treat $\phi$ as representing the trajectory effects and $U$ as representing the target effects.  To solve a sonar classification problem, all that matters is the signature: the $U$ function.  Two distinct sonar collections $u_1 = U \circ \phi_1$ and $u_2 = U \circ \phi_2$ with identical signatures should therefore be considered equivalent for the purposes of classification.  We capture this idea with a definition of \emph{signature equivalence} for factorizations.  This section explores what properties are entailed when two collections have equivalent signatures.

\begin{definition}
  If $u_1 = U \circ \phi_1$ and $u_2 = U \circ \phi_2$ are quasiperiodic factorizations such that the diagram
\begin{equation}
  \label{eq:quasi-factor-equivalent}
  \xymatrix{
M_1 \ar[d]_{\phi_1} \ar[rd]^{u_1} &         \\
C \ar[r]_{U}&N                       \\
M_2 \ar[u]^{\phi_2} \ar[ru]_{u_2}& \\
    }
\end{equation}
  commutes, we say that the factorizations of $u_1$ and $u_2$ have \emph{equivalent signatures}.  If instead of quasiperiodic factorizations in the diagram \eqref{eq:quasi-factor-equivalent}, we chose to use constant rank factorizations, then we say that the factorizations of $u_1$ and $u_2$ have \emph{constant rank equivalent signatures}.
\end{definition}

While equivalent signatures can also be thought of as a feature of the \emph{functions} $u_1$ and $u_2$ (rather than the factorizations), constant rank signature equivalence becomes trivial when thought of this way (Proposition \ref{prop:crse_existence}).

If $u_1 = U \circ \phi_1$ and $u_2 = U \circ \phi_2$ have equivalent signatures, one might think that the categories of quasiperiodic factorizations $\cat{QuasiP}(u_1)$ and $\cat{QuasiP}(u_2)$ are equivalent.  This is not the case, as demonstrated by the next example.

\begin{example}
  \label{eg:sig_not_equivalent}
  Consider $M_1=M_2 =N = S^1$, the unit circle parameterized by an angle $\theta$.  Suppose that
  \begin{equation*}
    u_1(\theta) = 2 \theta,
  \end{equation*}
  and
  \begin{equation*}
    u_2(\theta) = 6 \theta,
  \end{equation*}
  where we assume that angles outside the range $[0,2\pi)$ are wrapped back into that range.  These functions have universal quasiperiodic factorizations with the same phase space, $C=S^1$, and $U=\id$ in both cases.  The diagram \eqref{eq:quasi-factor-equivalent} becomes
    \begin{equation*}
  \xymatrix{
S^1 \ar[d]_{u_1} \ar[rd]^{u_1} &         \\
S^1 \ar[r]_{\id}&S^1                       \\
S^1 \ar[u]^{u_2} \ar[ru]_{u_2}& \\
    }
    \end{equation*}
    Thus, these two quasiperiodic factorizations have equivalent signatures, but according to Theorem \ref{thm:circle_quasip}, the categories $\cat{QuasiP}(u_1)$ and $\cat{QuasiP}(u_2)$ are not equivalent, because $2$ is prime and $6$ is composite.
\end{example}

Calling a signature equivalence an \emph{equivalence} is legitimate, because it is in fact an equivalence relation.

\begin{proposition}
  Signature equivalence is an equivalence relation between two functions $u_1,u_2: M \to N$.  This remains true for constant rank signature equivalence.
\end{proposition}
\begin{proof}
  Symmetry is immediate from the definition.  Reflexivity follows immediately upon recognizing that every smooth function has a trivial quasiperiodic factorization.  To attempt to establish transitivity, the diagram we start with is
\begin{equation*}
  \xymatrix{
M_1 \ar[d]_{\phi_1} \ar[rd]^{u_1} & &        \\
C \ar[r]^{U}&N&                       \\
M_2 \ar[u]^{\phi_2} \ar[ru]_{u_2} \ar[r]_{\phi'_2}&  C' \ar[u]_{U'} & M_3 \ar[l]^{\phi_3} \ar[lu]_{u_3}\\
    }
\end{equation*}
and what we want to construct is $C''$ so that the diagram
\begin{equation*}
  \xymatrix{
M_1 \ar[d]_{\phi_1} \ar[rd]^{u_1} \ar[rr]^{\phi''_1} & &       C'' \ar[dl] \\
C \ar[r]^{U}&N&                       \\
M_2 \ar[u]^{\phi_2} \ar[ru]_{u_2} \ar[r]_{\phi'_2}&  C' \ar[u]_{U'} & M_3 \ar[l]^{\phi_3} \ar[lu]_{u_3} \ar[uu]_{\phi''_3}\\
    }
\end{equation*}
commutes with surjective submersions $\phi''_1$ and $\phi''_3$.  This can be accomplished by the unique\footnote{This is excessive; \cite[Lem. 14]{Robinson_SampTA_2015} is really all that is needed here.} \emph{universal quasiperiodic factorization} \cite[Thm. 5]{Robinson_SampTA_2015} (or using Proposition \ref{prop:universal_constant_rank} in the case of constant rank factorizations) for $M_2$, since this means that we can expand the corresponding portion of the diagram 
\begin{equation*}
  \xymatrix{
    C \ar[rr]^{U} \ar[dr]_{c} && N \\
    &C''\ar[ur]^{U''}&\\
    M_2 \ar[uu]^{\phi_2} \ar[rr]_{\phi'_2} \ar[ur]_{\phi''_2} && C' \ar[uu]_{U'} \ar[ul]^{c'}\\
    }
\end{equation*}
to include such a factorization of $u_2 : M_2 \to N$ into $U'' \circ \phi''_2$.  Composing maps from this diagram and the previous, we have the diagram
\begin{equation*}
  \xymatrix{
    M_1 \ar[d]_{\phi_1} \ar[rd]^{u_1} & &     \\
C \ar[r]^{U} \ar[d]_{c}&N&                       \\
C''  \ar[ru]^{U''} &  C' \ar[u]_{U'} \ar[l]^{c'} & M_3 \ar[l]^{\phi_3} \ar[lu]_{u_3} \\
    }
\end{equation*}
which we can use to define $\phi''_1 = c \circ \phi_1$ and $\phi''_3 = c' \circ \phi_3$.  That these two maps are surjective submersions is a consequence of \cite[Lem. 14]{Robinson_SampTA_2015}, which completes the argument.
\end{proof}

What do equivalence classes of signature equivalent functions look like?  For one, all members of an equivalence class share the same image set.  This provides a rather direct explanation of why classifying sonar targets using the persistent homology of the space of echoes is effective \cite{sonarspace}.

\begin{proposition}
  \label{prop:identical_images}
  If $u_1$ and $u_2$ are signature equivalent functions, then they have identical image sets.
\end{proposition}
This need not be true for constant rank signature equivalence.
\begin{proof}
  Since $u_1$ and $u_2$ are signature equivalent, there is a pair of quasiperiodic factorizations, $u_1 = U \circ \phi_1$ and $u_2 = U \circ \phi_2$, where $U : C \to N$ for some space $C$.  Suppose that $y \in N$ is in the image of $u_1$, which means $y$ is in the image of $U$ as well.  That means there is an $x \in C$ such that $U(x) = y$.  Since $\phi_2$ is surjective by assumption, $x$ is in the image of $\phi_2$, and hence there is a $z$ for which $u_2(z) = U(\phi_2(z)) = U(x) = y$.  This establishes that $\image u_2 \subseteq \image u_1$.

  On the other hand, since $\phi_1$ is also surjective by assumption, then a similar argument establishes that $\image u_1 \subseteq \image u_2$.
\end{proof}

Like the functoriality expressed in Propositions \ref{prop:left_functor} and \ref{prop:right_functor}, (constant rank) signature equivalence respects pre- and post-composition with smooth functions.

\begin{proposition}
  \label{prop:crse_functoriality_prep}
  Suppose that $u_1: M_1 \to N$ and $u_2 : M_2 \to N$ have equivalent signatures.
  \begin{enumerate}
  \item If $f: N \to N'$ is a smooth map, then $(f \circ u_1)$ and $(f \circ u_2)$ have equivalent signatures.
  \item If $g: M' \to M_1$ is a surjective submersion, then $(u_1 \circ g)$ and $u_2$ have equivalent signatures.
  \item Statements (1) and (2) hold for constant rank versions, \emph{mutatis mutandis}.
  \end{enumerate}
\end{proposition}
\begin{proof}
  All of the statements implied by this proposition follow from reasoning about a diagram of the form
  \begin{equation}
  \xymatrix{
M' \ar[r]^g \ar@/_1pc/[dr]_{\phi_1 \circ g}&M_1 \ar[d]_{\phi_1} \ar[rd]^{u_1} \ar@/^1pc/[rrd]^{f \circ u_1} &         \\
&C \ar[r]_{U}&N \ar[r]^f& N'                       \\
&M_2 \ar[u]^{\phi_2} \ar[ru]_{u_2} \ar@/_1pc/[rru]_{f \circ u_2}& \\
    }
\end{equation}
\end{proof}

Finally, although Example \ref{eg:sig_not_equivalent} shows that signature equivalence does not ensure that categories of quasiperiodic factorizations are equivalent, the categories are related in a more subtle way.  Intuitively, the presence of a signature equivalence between two functions sets up an equivalence between subcategories within their respective categories of quasiperiodic factorizations.  To state this precisely requires the notion of a \emph{coslice category}.

\begin{definition}(Standard)
  If $A$ is an object of a category $\cat{C}$, the \emph{coslice category $(A \downarrow \cat{C})$} contains every object $B$ in $\cat{C}$ for which there is a morphism $A \to B$.  Morphisms of $(A \downarrow \cat{C})$ consist of commuting diagrams of the form
  \begin{equation*}
    \xymatrix{
      &A\ar[dr]\ar[dl]&\\
      B \ar[rr]_f && B'
      }
  \end{equation*}
  where $f$ is a morphism of $\cat{C}$.
\end{definition}

Briefly, the coslice category $(A \downarrow \cat{C})$ is equivalent to the subcategory of $\cat{C}$ generated by all morphisms and objects ``downstream'' of $A$.

\begin{theorem}
  \label{thm:coslice_equivalent}
  Suppose that $u_1 : M_1 \to N$ and $u_2: M_2 \to N$ are two maps with equivalent signatures.  Suppose that we write the corresponding quasiperiodic factorizations $u_1 = U \circ \phi_1$ and $u_2 = U \circ \phi_2$, respectively.  Then the coslice categories $(U \circ \phi_1 \downarrow \cat{QuasiP}(u_1))$ and $(U \circ \phi_2 \downarrow \cat{QuasiP}(u_2))$ are isomorphic (not merely equivalent). 
\end{theorem}

The proof of Theorem \ref{thm:coslice_equivalent} relies upon an explicit construction of a functor between the coslice categories, the mechanics of which are explained in Lemma \ref{lem:sig_equiv_functor} below.

\begin{lemma}
  \label{lem:sig_equiv_functor}
  Suppose that $u_1 : M_1 \to N$ and $u_2: M_2 \to N$ are two maps with quasiperiodic factorizations $u_1 = U \circ \phi_1$ and $u_2 = U \circ \phi_2$.  If $u_1 = U' \circ \phi_1'$ is another quasiperiodic factorization of $u_1$ such that there is a $\cat{QuasiP}(u_1)$ morphism
    \begin{equation*}
    \xymatrix{
      M_1 \ar[rr]^{u_1} \ar[dr]^{\phi_1} \ar@/_1pc/[ddr]_{\phi_1'} && N \\
      &C \ar[d]_-{c} \ar[ur]^{U}&\\
      &C' \ar@/_1pc/[uur]_{U'}&\\
      }
  \end{equation*}
    where $c: C_1 \to C_2$ is a smooth map,
    then this induces another quasiperiodic factorization of $u_2=U' \circ \phi_2'$ and a $\cat{QuasiP}(u_2)$ morphism
     \begin{equation*}
    \xymatrix{
      M_2 \ar[rr]^{u_2} \ar[dr]^{\phi_2} \ar@/_1pc/[ddr]_{\phi_2'} && N \\
      &C \ar[d]_-{c} \ar[ur]^{U}&\\
      &C' \ar@/_1pc/[uur]_{U'}&\\
      }
     \end{equation*}
     The statement remains true of constant rank factorizations as well.
\end{lemma}
\begin{proof}
  This Lemma hinges on the observation that $\phi_1' = c \circ \phi_1$.  We can then simply make the definition $\phi_2'=c \circ \phi_2$, because all of the various paths in the diagram for $u_2$
  \begin{equation*}
    \xymatrix{
      M_1 \ar[rr]^{u_1} \ar[dr]^{\phi_1} \ar@/_1pc/[ddr]_{\phi_1'} && N \\
      &C \ar[d]_-{c} \ar[ur]^{U}&\\
      &C' \ar@/_1pc/[uur]^-(0.7){U'}|(0.4)\hole &&M_1 \ar[ll]^{\phi_2'} \ar[uul]_{u_2} \ar[ull]_{\phi_2} \\
      }
  \end{equation*}
  continue to commute.  Notice that commutativity of the original morphism diagram ensures that $c$ be a surjective submersion, so it follows that $\phi_2'$ is also a surjective submersion.
  The statement still holds for constant rank maps for essentially the same reason: it follows from the definition of morphism that $c$ must be a constant rank map.
\end{proof}

With Lemma \ref{lem:sig_equiv_functor} in hand, we can prove Theorem \ref{thm:coslice_equivalent}.

\begin{proof}(of Theorem \ref{thm:coslice_equivalent})
  Observe that Lemma \ref{lem:sig_equiv_functor} defines a function on objects:
  \begin{equation*}
    F : (U \circ \phi_1 \downarrow \cat{QuasiP}(u_1)) \to (U \circ \phi_2 \downarrow \cat{QuasiP}(u_2)).
  \end{equation*}
  Dually, the same construction works to define a function
  \begin{equation*}
    G : (U \circ \phi_2 \downarrow \cat{QuasiP}(u_2)) \to (U \circ \phi_1 \downarrow \cat{QuasiP}(u_1)).
  \end{equation*}

  We need to establish that (1) $F$ respects morphisms and therefore defines a covariant functor, (2) that $F \circ G = \id$, and that (3) $G \circ F = \id$.  Evidently due to the symmetry of the situation, (2) and (3) can be established simultaneously.

  Suppose that we have a morphism
    $m: (\phi_1',U') \to (\phi_1'',U'')$ in $(U \circ \phi_1 \downarrow \cat{QuasiP}(u_1))$,
  which means that the following diagram commutes
\begin{equation*}
    \xymatrix{
      M_1 \ar[rr]^{u_1} \ar[dr]^{\phi_1} \ar@/_1pc/[ddr]_{\phi_1'} \ar@/_2pc/[dddr]_{\phi_1''} && N \\
      &C \ar[d]_-{c'} \ar@/^2pc/[dd]^-{c''} \ar[ur]^{U}&\\
      &C' \ar@/_1pc/[uur]_{U'}|(0.27)\hole\ar[d]_-{m}&\\
      &C'' \ar@/_2pc/[uuur]_{U''}&\\
      }
\end{equation*}
where $c'$, $c''$, and $m$ are smooth maps (abusing notation by conflating the morphisms with their component maps on the phase spaces).
Notice in particular that $c''=m \circ c'$.
Lemma \ref{lem:sig_equiv_functor} uses the above to construct two diagrams for factorizations of $u_2$, namely
     \begin{equation*}
    \xymatrix{
      M_2 \ar[rr]^{u_2} \ar[dr]^{\phi_2} \ar@/_1pc/[ddr]_{\phi_2'} && N \\
      &C \ar[d]_-{c'} \ar[ur]^{U}&\\
      &C' \ar@/_1pc/[uur]_{U'}&\\
      }
     \end{equation*}
     and
    \begin{equation*}
    \xymatrix{
      M_2 \ar[rr]^{u_2} \ar[dr]^{\phi_2} \ar@/_1pc/[ddr]_{\phi_2''} && N \\
      &C \ar[d]_-{c''} \ar[ur]^{U}&\\
      &C'' \ar@/_1pc/[uur]_{U''}&\\
      }
     \end{equation*}
    But since $c''=m \circ c'$, these two diagrams can be combined into
    \begin{equation*}
    \xymatrix{
      M_2 \ar[rr]^{u_2} \ar[dr]^{\phi_2} \ar@/_1pc/[ddr]_{\phi_2'} \ar@/_2pc/[dddr]_{\phi_2''} && N \\
      &C \ar[d]_-{c'} \ar@/^2pc/[dd]^-{c''} \ar[ur]^{U}&\\
      &C' \ar@/_1pc/[uur]_{U'}|(0.27)\hole\ar[d]_-{m}&\\
      &C'' \ar@/_2pc/[uuur]_{U''}&\\
      }
    \end{equation*}
    which is the diagram of a morphism
      $F(m) : (\phi_2',U') \to (\phi_2'',U'')$ in $(U \circ \phi_2 \downarrow \cat{QuasiP}(u_2))$.
    Incidentally $F(m)$ is also given by the smooth map $m$ on the phase spaces.

    To establish that $F$ and $G$ are inverses, it is merely necessary to compute using the recipe from Lemma \ref{lem:sig_equiv_functor}:
    \begin{eqnarray*}
      G(F((\phi_1',U'))) &=& G(F((c \circ \phi_1, U'))) \\
      &=& G((c \circ \phi_2, U'))\\
      &=& (c \circ \phi_1, U')\\
      &=& (\phi_1', U')
    \end{eqnarray*}
    as desired.
\end{proof}

If the factorizations of two functions $u_1$ and $u_2$ are signature equivalent and are final elements in their respective categories, then their coslice categories are rather small and uninformative.  Conversely (though trivially!), if the factorizations are the trivial ones, then this forces their respective categories to be isomorphic.

\subsection{Constant rank signature equivalences}

\emph{Constant rank} signature equivalences appear to have many of the same features of signature equivalences, since several of the proofs from the statements in Section \ref{sec:comparing} carry over to constant rank signature equivalences without much effort.  However, constant rank signature equivalences are a relatively coarse way of comparing two functions, because there are simply too many of them, as will be shown in Proposition \ref{prop:crse_existence}.  Intuitively, this means that sonar signatures arising in this case have too many degrees of freedom to be effectively discriminated.  This is a bit disappointing because constant rank signature equivalences are closely related to homotopies (as will be shown in Proposition \ref{prop:crse_homotopy}) even though signature equivalences are not.  The solution appears in Definition \ref{def:crse}.  Instead of considering individual equivalences, we should consider a category of \emph{all possible such} constant rank equivalences.  This idea builds over the next few sections, culminating in Theorem \ref{thm:crse_coslice}, which characterizes the category of constant rank equivalences in terms of coslices.

A close examination reveals that none of the arguments in the proof of Theorem \ref{thm:coslice_equivalent} were dependent on the surjectivity of the phase maps.  Therefore, we have the following Corollary.

\begin{corollary}
  \label{cor:coslice_equivalent}
  Suppose that $u_1 : M_1 \to N$ and $u_2: M_2 \to N$ are two maps with constant rank factorizations $u_1 = U \circ \phi_1$ and $u_2 = U \circ \phi_2$ being constant rank signature equivalent factorizations.  Then the coslice categories $(U \circ \phi_1 \downarrow \cat{Const}(u_1))$ and $(U \circ \phi_2 \downarrow \cat{Const}(u_2))$ are isomorphic (not merely equivalent).   
\end{corollary}

Constant rank signature equivalences always exist between functions, at least in a trivial form, as the next Proposition shows.  

\begin{proposition}
  \label{prop:crse_existence}
  Suppose that $u_1: M_1 \to N$ and $u_2: M_2 \to N$ are two smooth functions.  Recalling that $M_1 \sqcup M_2$ is the disjoint union of the two manifolds involved, then
  \begin{equation*}
      \xymatrix{
M_1 \ar[d]_{i_1} \ar[rd]^{u_1} &         \\
M_1\sqcup M_2 \ar[r]^{U}&N                       \\
M_2 \ar[u]^{i_2} \ar[ru]_{u_2}& \\
    }
\end{equation*}
  is a constant rank signature equivalence when the vertical $i_k$ maps are inclusions and
  \begin{equation*}
    U(x) = \begin{cases}
      u_1(x) & \text{if }x\in M_1 \subseteq M_1 \sqcup M_2, \\
      u_2(x) & \text{if }x\in M_2 \subseteq M_1 \sqcup M_2. \\
    \end{cases}
  \end{equation*}
\end{proposition}
\begin{proof}
  One merely needs to recognize that the diagram commutes by construction, and that inclusions are automatically constant rank maps.
\end{proof}

The reader is cautioned from reading too much into Proposition \ref{prop:crse_existence}.  Because of Proposition \ref{prop:initial}, neither of the constant rank factorizations in Proposition \ref{prop:crse_existence} are initial objects in their categories of constant rank factorizations.  If the factorizations in Proposition \ref{prop:crse_existence} were initial objects in general, then Corollary \ref{cor:coslice_equivalent} would indicate that $\cat{Const}(u)$ is dependent \emph{only} upon the codomain of $u$.  The truth is quite a bit more subtle!

\begin{proposition}
  \label{prop:crse_homotopy}
  Suppose that $u_1, u_2 : M \to N$ are two maps and that $h : M \times [0,1] \to N$ is a smooth homotopy between them.  Then these maps are constant rank signature equivalent.  
\end{proposition}
\begin{proof}
  The hypothesis means that
\begin{equation*}
  \xymatrix{
    M \ar[d]_{i_0} \ar[dr]^{u_1}&\\
    M \times [0,1] \ar[r]^-{h} & N\\
    M \ar[u]^{i_1} \ar[ur]_{u_2}&\\
    }
\end{equation*}
commutes.  Furthermore, $u_1= h \circ i_0$ and $u_2 = h \circ i_1$ are both constant rank factorizations\footnote{Neither are surjective, so this does not result in quasiperiodic factorizations}.  
\end{proof}

Therefore, the constant rank signature equivalence classes contain entire homotopy classes of maps.

\begin{proposition}
  \label{prop:crse_diffeo}
Suppose $M$ is a connected manifold, that $u_1,u_2: M \to N$ are smooth maps, and that $u_2 = u_1 \circ f$ where $f: M \to M$ is a diffeomorphism.  Then these maps are constant rank signature equivalent.  
\end{proposition}
\begin{proof}
    The hypothesis establishes that
  \begin{equation*}
    \xymatrix{
      M \ar[dr]^{u_1} \ar[d]_{\id_M} & \\
      M \ar[r]^{u_1} & N\\
      M \ar[u]^f \ar[ur]_{u_2}
      }
  \end{equation*}
  commutes; evidently $f$ and $\id_M$ are of constant rank.
\end{proof}

Using Proposition \ref{prop:crse_diffeo} to establish a constant rank signature equivalence along with Corollary \ref{cor:coslice_equivalent} provides an alternative proof of Theorem \ref{thm:crse_diffeo} since the factorizations in question are evidently initial according to Proposition \ref{prop:initial}.

\subsection{Categories \emph{of} signature equivalences}
The fact that constant rank signature equivalences are so common suggests that we abstract further and consider these equivalences as objects in their own right.  

\begin{definition}
  \label{def:crse}
The category $\cat{CRSE}(u_1,u_2)$ for each pair of smooth functions $u_1: M_! \to N$, $u_2: M_2 \to N$ has constant rank signature equivalences as its objects.  Specifically, the objects are diagrams of the form
\begin{equation}
    \xymatrix{
M_1 \ar[d]_{\phi_1} \ar[rd]^{u_1} &         \\
C \ar[r]_{U}&N                       \\
M_2 \ar[u]^{\phi_2} \ar[ru]_{u_2}& \\
    }
\end{equation}
in which $\phi_1$ and $\phi_2$ are constant rank maps.  A morphism in $\cat{CRSE}(u_1,u_2)$ is determined by a map $c: C \to C'$ such that the diagram
\begin{equation}
  \xymatrix{
      M_1 \ar[rr]^{u_1} \ar[dr]^{\phi_1} \ar@/_1pc/[ddr]_{\phi_1'} && N \\
      &C \ar[d]_-{c} \ar[ur]^{U}&\\
      &C' \ar@/_1pc/[uur]^-(0.7){U'}|(0.4)\hole &&M_2 \ar[ll]^{\phi_2'} \ar[uul]_{u_2} \ar[ull]_{\phi_2} \\
      }
\end{equation}
commutes.
\end{definition}

Intuitively, each object of $\cat{CRSE}(u_1,u_2)$ defined above is an individual sonar target that could have yielded both signals $u_1$ and $u_2$ under different sensor conditions.  Therefore, a small number of isomorphism classes of $\cat{CRSE}(u_1,u_2)$ suggests that $u_1$ and $u_2$ are likely from different targets, while a larger number of isomorphism classes means that it is likely that $u_1$ and $u_2$ arise from similar targets.

\begin{corollary}
  \label{cor:crse_projection}
  There is a forgetful functor $\cat{CRSE}(u_1,u_2) \to \cat{Const}(u_1)$ that acts by truncating the diagram on the left to produce the one on the right:
  \begin{equation*}
    \xymatrix{
      M_1 \ar[d]_{\phi_1} \ar[rd]^{u_1} & &&& M_1 \ar[d]_{\phi_1} \ar[rd]^{u_1}       \\
C \ar[r]_{U}&N & \ar@{=>}[r]&&   C \ar[r]_{U}&N                    \\
M_2 \ar[u]^{\phi_2} \ar[ru]_{u_2}&& \\
      }
  \end{equation*}
  Evidently there is also a forgetful functor $\cat{CRSE}(u_1,u_2) \to \cat{Const}(u_2)$ that acts by truncating the top portion of the diagram instead.
\end{corollary}

\begin{theorem}
  \label{thm:surjectivity}
  Suppose that $u_1$ and $u_2$ are smooth maps $M \to N$ for a connected manifold $M$, and that $u_2 = u_1 \circ f$ for some diffeomorphism $f$.  Then the forgetful functors defined in Corollary \ref{cor:crse_projection} are surjective on objects (but not necessarily morphisms).
\end{theorem}
\begin{proof}
  Suppose that $u_1 = U \circ \phi$ is a constant rank factorization of $u_1$.
  By hypothesis, we therefore have the commutative diagram
  \begin{equation*}
    \xymatrix{
      M \ar[d]_\phi \ar[dr]^{u_1} & M \ar[l]_f \ar[d]^{u_2}\\
      C \ar[r]_U & N
      }
  \end{equation*}
  which can be rearranged as
  \begin{equation*}
    \xymatrix{
      M \ar[d]_{\phi} \ar[dr]^{u_1} & \\
      C \ar[r]_U & N \\
      M \ar[u]^{\phi \circ f} \ar[ur]_{u_2}\\
      }
  \end{equation*}
  which is evidently an object in $\cat{CRSE}(u_1,u_2)$ whose image through the forgetful functor is $U \circ \phi$.  Repeating the argument using $u_1 = u_2 \circ \phi^{-1}$ completes the proof.
\end{proof}

\begin{example}
  \label{eg:surjectivity}
  Let $u_1: \mathbb{R} \to \mathbb{R}$ be the constant map $u_1(x) = 0$ and $u_2: \mathbb{R} \to \mathbb{R}$ be the identity map $u_2(x) = x$.  The functor in Corollary \ref{cor:crse_projection} for $\cat{CRSE}(u_1,u_2) \to \cat{Const}(u_1)$ is not surjective on objects.  To see this, notice that the constant rank factorization $\mathbb{R} \to \star \to \mathbb{R}$ of $u_1$ through the single-point space $\star$ cannot be constant rank signature equivalent to any constant rank factorization of $u_2$, and thus must be outside the image of the functor.
\end{example}

Taking Theorem \ref{thm:surjectivity} and Example \ref{eg:surjectivity} together, this gives a means for comparing two smooth maps $u_1$, $u_2$. If the images of $\cat{CRSE}(u_1,u_2)$ in $\cat{Const}(u_1)$ and $\cat{Const}(u_2)$ are both large, then the maps are similar.  Homotopies yield another, related notion of similarity between maps as the next example shows.

\begin{example}
  If $u_1, u_2: M \to N$ are homotopic smooth maps, then the combined effect of Propositions \ref{prop:crse_existence} and \ref{prop:crse_homotopy} is that there is a morphism in $\cat{CRSE}(u_1,u_2)$ of the form
  \begin{equation*}
      \xymatrix{
      M \ar[rr]^{u_1} \ar[dr]^{I_0} \ar@/_1pc/[ddr]_{i_0} && N \\
      &M \sqcup M \ar[d]_-{c} \ar[ur]^{U}&\\
      &M \times [0,1] \ar@/_1pc/[uur]^-(0.7){h}|(0.42)\hole &&M \ar[ll]^{i_1} \ar[uul]_{u_2} \ar[ull]_{I_1} \\
      }
  \end{equation*}
  where the $I_k$ and $i_k$ maps are the obvious inclusions, $U$ is defined as in Proposition \ref{prop:crse_existence}, and
  \begin{equation*}
    c(x) = \begin{cases}
      (x, 0) & \text{if }x \text{ is in the first copy of }M \\
      (x, 1) & \text{if }x \text{ is in the second copy of }M.
    \end{cases}
  \end{equation*}
\end{example}

Given the definition of $\cat{CRSE}(u_1,u_2)$, it is easy to recast some of the previous results on functoriality.

\begin{proposition}
  \label{prop:crse_functoriality}
  Suppose that $u_1:M_1 \to N$, $u_2: M_2\to N$ and $f: N \to N'$ are smooth maps.  Then $f$ induces a covariant functor $\cat{CRSE}(u_1,u_2) \to \cat{CRSE}(f \circ u_1,f\circ u_2)$ via the construction in statement (1) of Proposition \ref{prop:crse_functoriality_prep}.
\end{proposition}
\begin{proof}
  This follows from a tedious (but easy) diagram chase using the diagram in the proof of Proposition \ref{prop:crse_functoriality_prep}.
  Covariance is assured for the same reason as in Proposition \ref{prop:left_functor}.
\end{proof}

\begin{proposition}
  \label{prop:crse_initial}
  The constant rank signature equivalence constructed in Proposition \ref{prop:crse_existence} for $u_1: M_1 \to N$ and $u_2: M_2 \to N$ is the initial object of $\cat{CRSE}(u_1,u_2)$.
\end{proposition}
\begin{proof}
  Suppose that we have an arbitrary object of $\cat{CRSE}(u_1,u_2)$, given by the diagram
\begin{equation*}
    \xymatrix{
M_1 \ar[d]_{\phi_1} \ar[rd]^{u_1} &         \\
C \ar[r]_{V}&N                       \\
M_2 \ar[u]^{\phi_2} \ar[ru]_{u_2}& \\
    }
\end{equation*}
Consider the map $c : (M_1 \sqcup M_2) \to C$ given by
\begin{equation*}
  c(x) = \begin{cases}
    \phi_1(x) & \text{if }x\in M_1 \subseteq M_1 \sqcup M_2 \\
    \phi_2(x) & \text{if }x\in M_2 \subseteq M_1 \sqcup M_2 \\
    \end{cases}
\end{equation*}
By construction, this map makes the diagram
\begin{equation*}
  \xymatrix{
      M_1 \ar[rr]^{u_1} \ar[dr]^{i_1} \ar@/_1pc/[ddr]_{\phi_1} && N \\
      &M_1 \sqcup M_2 \ar[d]_-{c} \ar[ur]^{U}&\\
      &C \ar@/_1pc/[uur]^-(0.7){V}|(0.4)\hole &&M_2 \ar[ll]^{\phi_2} \ar[uul]_{u_2} \ar[ull]_{i_2} \\
      }
\end{equation*}
commute when the $i_k$ maps are the inclusions.
\end{proof}

The initial objects of $\cat{CRSE}(u_1,u_2)$ and that of $\cat{Const}(u_1)$ are quite incompatible, even if there is a diffeomorphism $u_2 = u_1 \circ f$!  This is because there is no map $c$ that will make the following diagram commute:
\begin{equation*}
    \xymatrix{
      M \ar[rr]^{u_1} \ar[dr]^{\id_M} \ar@/_1pc/[ddr]_{i_1} && N \\
      &M \ar[d]_-{c} \ar[ur]^{u_1}&\\
      &M \sqcup M \ar@/_1pc/[uur]^-(0.7){U}|(0.4)\hole &&M \ar[ll]^{i_2} \ar[uul]_{u_2} \ar[ull]_{f} \\
      }
\end{equation*}
Moreover, considering the images of these objects of $\cat{CRSE}(u_1,u_2)$ in $\cat{Const}(u_1)$, there is precisely one $c$ that will make the following diagram commute:
\begin{equation*}
  \xymatrix{
    M \ar[rr]^{u_1} \ar[dr]^{\id_M} \ar@/_1pc/[ddr]_{i_1} && N \\
    &M \ar[d]_-{c} \ar[ur]^{u_1}&\\
    &M \sqcup M \ar@/_1pc/[uur]^-(0.7){U} \\
    }
\end{equation*}
Reversing the direction of the $c$ maps in the above two diagrams yields exactly one possible option in the top diagram, and many options in the bottom one. 

On the other hand, final objects of $\cat{CRSE}(u_1,u_2)$ and  $\cat{Const}(u_1)$ are compatible.

\begin{proposition}
  Suppose that
  \begin{equation*}
    \xymatrix{
M_1 \ar[d]_{\phi_1} \ar[rd]^{u_1} &         \\
C \ar[r]^{U}&N                       \\
M_2 \ar[u]^{\phi_2} \ar[ru]_{u_2}& \\
    }
  \end{equation*}
  is a constant rank signature equivalence and that $u_1 = U \circ \phi_1$ and $u_2 = U \circ \phi_2$ are final objects in $\cat{Const}(u_1)$ and $\cat{Const}(u_2)$, respectively.  Then their constant rank signature equivalence is final in $\cat{CRSE}(u_1,u_2)$.
\end{proposition}
\begin{proof}
  Suppose that
  \begin{equation*}
    \xymatrix{
      M_1 \ar[d]_{\psi_1} \ar[rd]^{u_1} &         \\
      C' \ar[r]^{V}&N                       \\
      M_2 \ar[u]^{\psi_2} \ar[ru]_{u_2}& \\
    }
  \end{equation*}
  is another object of $\cat{CRSE}(u_1,u_2)$.  Using the hypothesis that $u_1 = U \circ \phi_1$ and $u_2 = U \circ \phi_2$ are final objects, this implies that there is a $\cat{Const}(u_1)$ morphism
  \begin{equation*}
    \xymatrix{
      M_1 \ar[rr]^{u_1} \ar[dr]^{\psi_1} \ar@/_1pc/[ddr]_{\phi_1} && N \\
      &C' \ar[d]_-{c_1} \ar[ur]^{V}&\\
      &C \ar@/_1pc/[uur]_{U}&\\
      }
  \end{equation*}
  and a $\cat{Const}(u_2)$ morphism
  \begin{equation*}
    \xymatrix{
      M_2 \ar[rr]^{u_2} \ar[dr]^{\psi_2} \ar@/_1pc/[ddr]_{\phi_2} && N \\
      &C' \ar[d]_-{c_2} \ar[ur]^{V}&\\
      &C \ar@/_1pc/[uur]_{U}&\\
      }
  \end{equation*}
  These can be assembled into a larger diagram
  \begin{equation*}
    \xymatrix{
      &M_1 \ar@/_1pc/[ldd]_{\psi_1} \ar@/^1pc/[rdd]^{u_1} \ar[d]_{\phi_1} &         \\
      &C\ar[dr]^{U} \\
      C' \ar[rr]^{V} \ar[ur]_{c_1} \ar[dr]^{c_2} &&N                       \\
      &C \ar[ur]^{U} \\
      &M_2 \ar@/^1pc/[luu]^{\psi_2} \ar@/_1pc/[ruu]_{u_2} \ar[u]_{\phi_2} & \\
      }
  \end{equation*}
  To complete the argument, we need to show that the two maps $c_1$, $c_2$ can be replaced with a single map $c: C' \to C$, even though $c_1$ and $c_2$ might disagree.  To that end, consider $(C\sqcup C)/\sim$, where $x \sim y$ if there is a $z \in C'$ such that $c_1(z) =x$ and $c_2(z)=y$.  This construction makes the diagram
  \begin{equation*}
    \xymatrix{
      &C\ar[d]^{U} \ar[dr]^{i_1} \\
      C' \ar[r]^{V} \ar[ur]^{c_1} \ar[dr]_{c_2} &N&(C\sqcup C)/\sim \ar[l]_-{U'}     \\
      &C \ar[u]^{U} \ar[ur]_{i_2} \\
    }
  \end{equation*}
  commute if $U'$ is given by $U$ on both copies of $C$ in $(C\sqcup C)/\sim$.

  Splicing the two previous diagrams together yields the observation that there is a morphism
  \begin{equation*}
    \xymatrix{
      M_1 \ar[rr]^{u_1} \ar[dr]^{\phi_1} \ar@/_1pc/[ddr]_{i_1\circ\phi_1} && N \\
      &C \ar[d]_-{i_1} \ar[ur]^{U}&\\
      &(C\sqcup C)/\sim \ar@/_1pc/[uur]_{U'}&\\
      }
  \end{equation*}
  Applying finality in $\cat{Const}(u_1)$ once more, we must conclude that the factorization through $(C\sqcup C)/\sim$ is $\cat{Const}(u_1)$-isomorphic to $U\circ\phi_1$.  A similar argument holds for the other factorization.  Therefore, we must have that $(C\sqcup C)/\sim$ is diffeomorphic to $C$; so that $c_1$ and $c_2$ can only disagree up to that diffeomorphism.
\end{proof}

A consequence of Proposition \ref{prop:crse_initial} is a characterization of $\cat{CRSE}(u_1,u_2)$ in terms of coslice categories of $\cat{Const}(u_1)$ and $\cat{Const}(u_2)$.

\begin{theorem}
  \label{thm:crse_coslice}
  If $u_1: M_1 \to N$ and $u_2: M_2 \to N$ are two smooth maps, $\cat{CRSE}(u_1,u_2)$ is isomorphic to the coslice category $(U \circ i_1 \downarrow \cat{Const}(u_1))$ (and therefore also to the coslice category $(U \circ i_2 \downarrow \cat{Const}(u_2))$), where $U$, $i_1$, and $i_2$ are defined as in Proposition \ref{prop:crse_existence}.
\end{theorem}
\begin{proof}
  Each object of $\cat{CRSE}(u_1,u_2)$ already implies the existence of two constant rank factorizations: $u_1 = U' \circ \phi_1$ and $u_2 = U' \circ \phi_2$.

  According to Proposition \ref{prop:crse_initial}, we have a $\cat{CRSE}(u_1,u_2)$ morphism from $(U \circ i_1, U \circ i_2)$ (the initial object) to $(U' \circ \phi_1, U' \circ \phi_2)$ (an arbitrary object), implying the existence of corresponding morphisms in $\cat{Const}(u_1)$ and $\cat{Const}(u_2)$.  Consequently, this means that the factorizations $u_1 = U' \circ \phi_1$ and $u_2 = U' \circ \phi_2$ can be considered objects in the corresponding coslice categories.  Conversely, the construction of the equivalence between the coslice categories involved in the proof of Theorem \ref{thm:coslice_equivalent} is precisely the same construction as that for $\cat{CRSE}(u_1,u_2)$.
\end{proof}

\section{The factorization categories under CSAS trajectory distortions}
\label{sec:application}

Now that the theoretical tools have been established in the previous sections, let us revisit the CSAS example from Section \ref{sec:motivation}.  We remind the reader that for the purposes of this application, if we distort the sonar sensor's trajectory there will be a distortion on both the look angle and the range.  In the most general setting, both of these effects can be captured by the phase map $\phi$.  In the specific case of CSAS echoes, the range distortion can be removed easily by time aligning each pulse separately, assuming the target has a large enough cross section.  Therefore, we will assume that the trajectory distortions only apply to the look angle.

\begin{figure}
\begin{center}
\includegraphics[height=2in]{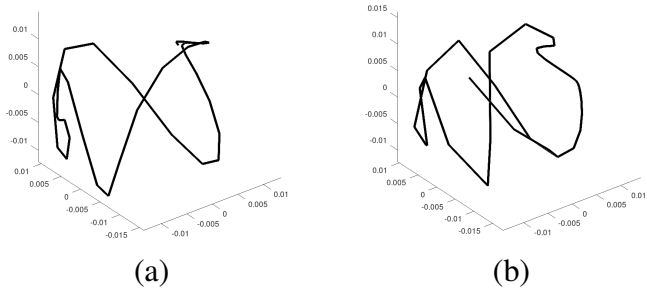}
\caption{Principal components analysis plots of the signature of (a) a $3$-fold symmetric scatterer (Figure \ref{fig:scatterer_sigs}(b)), (b) a distorted $3$-fold symmetric scatterer (Figure \ref{fig:scatterer_sigs2}(a)).}
\label{fig:scatterer_pca_compare}
\end{center}
\end{figure}

Recall that signatures from a $3$-fold symmetric scatterer without distortions (Figure \ref{fig:scatterer_sigs}(b)) and with distortions (Figure \ref{fig:scatterer_sigs2}(a)) were shown in Section \ref{sec:motivation}.  Since they differ by a smooth trajectory distortion, which happens to be a diffeomorphism, Theorem \ref{thm:crse_diffeo} argues that the factorization categories for these two signatures are isomorphic, not just equivalent, which means that there should be a bijective correspondence between their factorizations.  This additionally means that they should be signature equivalent.  Hence, Proposition \ref{prop:identical_images} implies that their image sets should be identical.  From a practical standpoint, this means that the set of possible echoes (rows) in Figure \ref{fig:scatterer_sigs}(b) and Figure \ref{fig:scatterer_sigs2}(a) should be identical.  A good way to see this is to compare their principal components analysis plots, shown in Figure \ref{fig:scatterer_pca_compare}.  In computing these plots, pairs of adjacent rows were used as coordinates for each pulse, and the axes (the principal vectors for Figure \ref{fig:scatterer_sigs}(b)) are the same for both plots.  Adjacent rows were used as a proxy for the derivative of the phase function, which is of constant rank by Proposition \ref{prop:signal_knot}.  Notice that the two frames of Figure \ref{fig:scatterer_pca_compare} show very similar trajectories; they only differ due to sampling effects because a finite number of pulses were used.

\begin{figure}
\begin{center}
\includegraphics[height=2in]{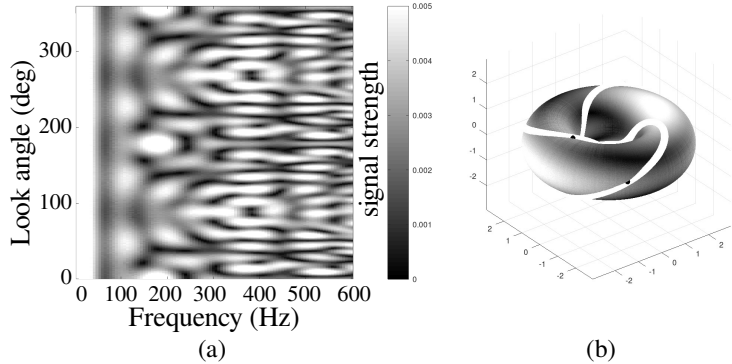}
\caption{(a) Signature of the sum of a $4$-fold symmetric scatterer and a $6$-fold scatterer, (b) Response as a function of scatterer angles at $300$ Hz.}
\label{fig:knot_4_6}
\end{center}
\end{figure}

In Section \ref{sec:motivation}, we considered the composite scatterer corresponding to a $(2,3)$ torus knot shown in Figure \ref{fig:scatterer_combined_sig}.  There are many other composite scatterers with the same image in the torus.  For instance, if we use $P=4$ and $Q=6$ in Equations \eqref{eq:one_scatterer_family} and \eqref{eq:two_scatterers}, we obtain a $(4,6)$ torus knot signature shown in Figure \ref{fig:knot_4_6}.  The image of the path (though not the path itself) in the torus is shown in Figure \ref{fig:knot_4_6}(b), and is the same as that of the $(2,3)$ torus knot shown in Figure \ref{fig:toroidal_flat}(a).

We can use Theorem \ref{thm:circle_quasip} to compare these two signatures, and it is the case that they can be distinguished by their corresponding $\cat{QuasiP}$ categories.  Notice that the torus has fundamental group $\pi_1(S^1 \times S^1) = \mathbb{Z} \times \mathbb{Z}$.  There are two cyclic subgroups of $\pi_1(S^1 \times S^1)$ that contain the $(4,6)$ torus knot, namely itself and the cyclic subgroup generated by the $(2,3)$ torus knot.  On the other hand, there is only one cyclic subgroup of $\pi_1(S^1 \times S^1)$ that contains the $(2,3)$ torus knot.  Thus the $\cat{QuasiP}$ categories of the $(2,3)$ and $(4,6)$ torus knots differ, and so there is no trajectory distortion that will transform one CSAS signature into the other.  We can therefore conclude that they must be different targets.

Leaving the scatterers fixed while changing the relative angle, changes the intercept of the trajectory along the torus.  Equivalently, changing the relative angle simply shifts the underlying function on the torus, as Figures \ref{fig:changes}(a)--(b) show.  These figures are therefore an example of constant rank signature equivalence.  On the other hand, if the symmetry of the scatterers is changed, this dramatically changes the function and the torus knot, as shown in Figure \ref{fig:changes}(a) and (c).

\begin{figure}
\begin{center}
\includegraphics[height=1.75in]{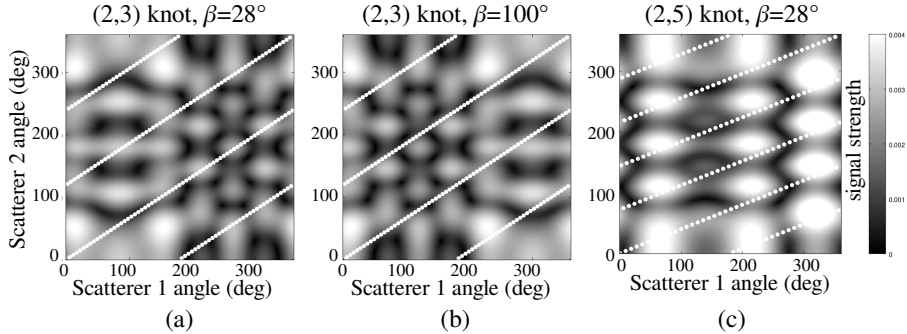}
\caption{Different responses at $300$ Hz for (a) the $(2,3)$ torus knot signature with relative angle $28^\circ$, (b) the $(2,3)$ torus knot signature with relative angle $100^\circ$, and (c) the $(2,5)$ torus knot signature with relative angle $28^\circ$.}
\label{fig:changes}
\end{center}
\end{figure}

Two targets with the same symmetry class (which have isomorphic $\cat{Const}$ categories), as in Figures \ref{fig:changes}(a)--(b), can be discriminated by constant rank signature equivalence since that is a generalization of comparing the values of their signatures on a fundamental domain.  In the case of changing the relative angles, these two targets are constant rank signature equivalent.  Figures \ref{fig:changes}(a) and (c) are not constant rank signature equivalent, since the torus functions (signatures) are different.  Evidently, classifying \emph{all} constant rank signature equivalences would involve classifying all functions on the torus.  This is not particularly easy to do completely, even though it is easy to identify when two such functions are equal.

\section{Conclusion}
\label{sec:conclusion}

This article has provided a nearly complete answer as to why topological methods are effective at solving sonar classification problems.
Specifically, topological methods decouple trajectory effects from the classification-relevant information contained in a sonar signature.
Because of this decoupling, topological methods can be used when trajectory information is inaccurate or missing.

In answering the practical challenge of understanding sonar classification, this article has defined several new topological invariants that apply to factorizations of smooth functions, and has established key properties about them.  In the case of the category $\cat{QuasiP}$, we have a complete characterization in the case that applies to circular synthetic aperture sonar (CSAS) classification problems.  This characterization relies on the computation of the fundamental group of the space of echoes, which can be estimated from sonar data using persistent homology.  Therefore, this article has provided theoretical justification for what was demonstrated experimentally in \cite{sonarspace}.  

Given the fact that several new invariants were discovered, much work still remains.
Although we understand how to compute the $\cat{QuasiP}$ category in some cases, we have only a very limited understanding of how to characterize the related category $\cat{Const}$.  The techniques that are effective for $\cat{QuasiP}$ are simply uninformative for $\cat{Const}$.
Additionally, the complete characterization of $\cat{QuasiP}$ only applies to a limited class of synthetic aperture problems, so further work to characterize it under arbitrary sonar collections still remains.

\section*{Acknowledgments}

 This article is based upon work supported by the Office of Naval Research (ONR) under Contract Nos. N00014-15-1-2090 and N00014-18-1-2541. Any opinions, findings and conclusions or recommendations expressed in this article are those of the author and do not necessarily reflect the views of the Office of Naval Research.
 
\bibliographystyle{plain}
\bibliography{constrank_bib}
\end{document}